\documentclass[12pt,a4paper]{article}
\usepackage[T1]{fontenc}

\usepackage{indentfirst,mathrsfs}
\usepackage{amsfonts,amsmath,amssymb,amsthm}
\usepackage{latexsym,amscd,multirow}
\usepackage{algorithm}
\usepackage{algpseudocode}
\usepackage{amsbsy}
\usepackage{xcolor}
\usepackage{diagbox}
\usepackage{colortbl}
\usepackage{multirow}
\usepackage{makecell}
\usepackage{float}
\usepackage{cases}
\usepackage{dsfont}
\usepackage{enumerate}
\usepackage{cite}
\usepackage{pifont}
\usepackage{graphicx}
\usepackage{mathtools}

\usepackage{tikz}
\usetikzlibrary{decorations,arrows}
\usetikzlibrary{decorations.markings}
\usetikzlibrary{decorations.pathreplacing}
\usetikzlibrary {datavisualization.formats.functions}

\usepackage{hyperref}
\hypersetup{hidelinks,
	backref=true,
	pagebackref=true,
	hyperindex=true,
	breaklinks=true,
	colorlinks=true,
	linkcolor=black!60,
	citecolor=black!50,
	urlcolor=blue,
	bookmarks=true,
	bookmarksopen=true
}
\usepackage{url}
\newtheorem{thm}{Theorem}
\newtheorem{lem}{Lemma}
\newtheorem{cor}{Corollary}
\newtheorem{Open problem}{Open Problem}

\newtheorem{prop}{Proposition}

\newtheorem{defn}{Definition}
\newtheorem{example}{Example}
\newtheorem{remark}{Remark}

\newcommand{\cn}{\color{black}}

\newcommand{\s}{\textbf{s}}
\newcommand{\sn}{\textbf{s}_n}
\newcommand{\an}{\textbf{a}_n}
\newcommand{\sni}{\textbf{s}_n^{\infty}}
\newcommand{\n}{\lfloor \frac{n}{2} \rfloor}
\newcommand{\nup}{\lceil \frac{n}{2} \rceil}
\newcommand{\nsmall}{\lceil\log_{2}(n)\rceil}
\newcommand{\calF}{\mathcal{F}}

\newcommand{\obr}[2]{\overbrace{#1 \dots #1}^{#2}}

\begin{document}
	\title{
		Further Investigations on Nonlinear Complexity of Periodic Binary Sequences
	}
	
	\author{
		Qin Yuan
		\thanks{Q. Yuan and X. Zeng are with Faculty of Mathematics and Statistics, Hubei Key Laboratory of
			Applied Mathematics, Hubei University, Wuhan 430062, Hubei, China. Email:
			yuanqin2020@aliyun.com, xzeng@hubu.edu.cn},
		Chunlei Li
		\thanks{C. Li and T. Helleseth are with the Department of Informatics, University of Bergen, Bergen, N-5020, Norway.
			Email: chunlei.li@uib.no, tor.helleseth@uib.no},
		Xiangyong Zeng,	
		Tor Helleseth,
		and
		Debiao He
		\thanks{D. He is with Key Laboratory of Aerospace Information Security and Trusted Computing Ministry of Education, School of Cyber Science and Engineering, Wuhan University, Wuhan 430072, Hubei, China. Email: hedebiao@whu.edu.cn}	
	}

	\date{}
	\maketitle
	
	\begin{abstract}	
		Nonlinear complexity is an important measure for assessing the randomness of sequences.
		In this paper we investigate how circular shifts affect the nonlinear complexities of finite-length binary sequences and
		then reveal a more explicit relation between nonlinear complexities of finite-length binary sequences and their corresponding periodic sequences.
		Based on the relation, 		
		we propose two algorithms that can generate all periodic binary sequences with any prescribed nonlinear complexity.
		
		{\small {\bf Index Terms:}} Periodic sequence, nonlinear complexity, randomness
	\end{abstract}
	
	\section{Introduction}
	Pseudorandom sequences have applications in various digital systems and communication technologies, such as radio communications, distance ranging, simulation, game theory, and cryptography \cite{Golomb2017}.
	The quality of randomness is a crucial factor for pseudorandom sequences in many of these applications, particularly for cryptographic applications.
	For assessing the randomness of pseudorandom sequences, different complexity measures were proposed in the literature 	 \cite{MartinLoef1966,Wang1988,Beth,Niederreiter1999}
	and the most well-understood one is probably the linear complexity. The linear complexity of a sequence is defined as the length of the shortest linear feedback shift registers (LFSRs)  that can generate the sequence \cite{MasseyShirlei1996_Crypto,Niederreiter2003}.
	Suppose a sequence $\s$ of length $n$ has linear complexity $l\leq n/2$. Given any of its $2l$-length subsequences, namely, $(s_{i}, s_{i+1}, \dots, s_{i+2l-1})$ for any $i\geq 0$,
	the Berlekamp-Massey algorithm \cite{Massey1969} can efficiently produce the linear recurrence of length $l$ and thereby the
	whole sequence $\s$.
	Hence pseudorandom sequences for cryptographic applications must not have low linear complexity.
	Rueppel \cite{Rueppel1986} conjectured that  $n$-periodic binary sequences have expected linear complexity close to $n$.
	Meidl and Niederreiter \cite{Meidl2002} confirmed this conjecture for arbitrary finite fields; in particular, they showed that $n$-periodic binary sequences have expected linear complexity at least $ \frac{3n-1}{4}$.
	Research has been done on the linear complexity of special sequences, for instance, Lempel-Cohn-Eastman sequences \cite{Helleseth}, Legendre sequences \cite{Ding1998}.
	Meanwhile, variants and extensions of linear complexity, e.g., linear complexity profile \cite{Niederreiter2003}, $k$-error linear complexity \cite{Meidl2002,Niederreiter2003}, quadratic span \cite{Rizo} and nonlinear complexity \cite{Jansen}
	have been studied.  Interested readers may refer to \cite{LWin,Chenzhix} for more discussions on these complexity measures.
	
	As an additional figure of merit to judge the randomness of a sequence, nonlinear complexity (also referred to as maximum-order complexity) was introduced by Jansen and Boekee in 1989 \cite{Jansen, JansenPhD}, where they
	defined it as the length of the shortest FSRs that generate a given sequence. Their work showed that the expected nonlinear complexity of random $n$-length $q$-ary sequences is approximately $2\log_q(n)$.
	Significant progress has been made on the derivation of the shortest feedback functions for a given sequence and the construction of sequences with high nonlinear complexity.
	Jansen and Boekee  \cite{Jansen}  initially related nonlinear complexity to the maximum depth of a directed
	acyclic word graph, which can be employed to determine the nonlinear complexity (profile) of a given sequence. Rizomiliotis and Kalouptsidis \cite{Rizomiliotis2005} in 2005 proposed an efficient algorithm, which exploited the special structure of associated linear equations, for finding the shortest feedback functions for a given sequence. Later Limniotis et al. \cite{LKK2} studied the relation between nonlinear complexity and Lempel-Ziv complexity, thereby presenting a recursive algorithm
	which has a similar procedure to the Berlekamp-Massey algorithm.
	As for the construction of sequences with high nonlinear complexity,
	researchers mainly exploited some algebraic tools and explored the structure of those constructed sequences \cite{PR,HX,LXY,Castellanos2022,PZS,YILIN,SZLH,liang,Xiao2018}.
	Niederreiter et al. \cite{HX}, Luo et al. \cite{LXY} and Castellanos et al. \cite{Castellanos2022} constructed sequences with high nonlinear complexity from function fields with many rational places.
	Based on detailed investigations of internal structures of sequences,
	all $q$-ary sequences of length $n$ and high nonlinear complexities $n-i, \, i=1,2,3,4$, were completely characterized in \cite{PZS, YILIN}.
	Very recently Liang et al. \cite{liang} completely determined $n$-length binary sequences with nonlinear complexity $\geq \frac{n}{2}$.

	Let $\sni$ be an $n$-periodic sequence over certain alphabet $\mathcal{A}$.
	Its nonlinear complexity is closely related to nonlinear complexities of its $n$-length subsequences, namely,
	$nlc(\sni) =nlc(\an^2)\geq nlc(\an)$, where $nlc$ denotes the nonlinear complexity of a given sequence, $\an$ is any $n$-length subsequence of $\sni$ and $\an^2=\an\an$ denotes the concatenation of two identical $\an$.
	As $\an$ can be chosen as any $n$-length subsequence of $\sni$, this inequality is relatively vague and does not reveal further insights into the relation between $nlc(\sni)$ and $nlc(\an)$.
	Only a few results have been reported on nonlinear complexity of periodic sequences so far.
	Rizomiliotis \cite{PR} exploited the power series representation of binary sequences and proposed two constructions of binary sequences of period $2^m-1$ with maximum nonlinear complexity for
	given linear complexity.
	In 2017 Sun et al. completely  characterized the structure of periodic sequences $\sni$ of maximum nonlinear complexity $n-1$ and proposed a recursive algorithm to generate all such sequences \cite{SZLH}. Later Xiao et al.
	\cite{Xiao2018} determined all binary sequences of period $n$ having nonlinear complexity $n-2$ in a similar way.
	Motivated by recent work 
	on nonlinear complexities of binary  $n$-length sequences \cite{liang},
	this paper aims to reveal more explicit relations between $nlc(\sni)$ and $nlc(\an)$, where $\an$ is a certain cyclic shift of  $\sn$.
	To this end, we are concerned with a particular set $\mathcal{B}(n,c)$ (in Def. \ref{BB}) of sequences with a specific structure.
	We start with investigating how circular shift operators affect 
	the nonlinear complexity of binary sequences $\sn$ in $\mathcal{B}(n,c)$.
	Then we study the relation between the companion pairs of $\sn$ and the companion pairs of $\sni$,
	which enables us to establish a one-to-one correspondence between periodic sequences $\sni$ with given nonlinear complexities and
	certain sequences $\sn$ in $\mathcal{B}(n,c)$ (see Thm. \ref{relation} and Thm. \ref{c+t}).
	Based on the correspondence, all periodic binary sequences $\sni$ with any prescribed nonlinear complexity $\omega$ can be completely generated (in Thm. \ref{thm_core}). The generation process is summarized in two algorithms (Alg.\ref{alg:period-small} and Alg.\ref{alg:period-large}), depending on the relation between $\omega=nlc(\sni)$ and $\frac{n}{2}$. 	
	To the best of our knowledge,  this is the first report of theoretical results on the subject
	that can be used to efficiently generate all periodic binary sequences with any prescribed nonlinear complexity.

	The remainder of this paper is organized as follows. 	
	In Section \ref{Sec2}, we introduce some basics of
 nonlinear complexity. 
	Section \ref{Sec2.2}  presents
	notations, definitions, and auxiliary results for the nonlinear complexity of  $n$-length binary sequences.
	Section \ref{Sec3} is dedicated to revealing more explicit
	relations between $nlc(\sni)$ and $nlc(\an)$, where $\an$ is a certain cyclic shift of $\sn$.
	In Section \ref{Sec4}, we propose two algorithms to generate periodic binary sequences with any prescribed nonlinear complexity. Finally,
	Section \ref{Sec6} concludes the work of this paper.

	\section{Preliminaries}\label{Sec2}
	 In Section \ref{Sec2} and Section \ref{Sec2.2}, we will introduce basic notations, definitions and auxiliary lemmas. For readers' convenience we summarise important notations in Table 1 at the end of Section \ref{Sec2.2}.

	For a positive integer $m$, an $m$-stage \textit{feedback shift register} (FSR) is a clock-controlled circuit consisting of $m$ consecutive storage units and a feedback function $f$ as displayed in Figure \ref{Fig1}.
	Starting with an initial state $\underline{\s}_0=(s_{0}, s_{1},\dots,s_{m-1})$, the states in the FSR will be updated by a clock-controlled transformation as follows:
	\begin{equation*}
		\calF: \underline{\s}_i = (s_{i},s_{i+1},\dots,s_{i+m-1}) \longmapsto   \underline{\s}_{i+1}=(s_{i+1},\dots,s_{i+m-1},s_{i+m}), \, i \geq 0,
	\end{equation*} where $s_{i+m}=f(s_{i},s_{i+1},\dots,s_{i+m-1})$, and the leftmost symbol  for each state $\underline{\s}_i$ will be output. In this way an FSR produces a sequence $\textbf{s}=(s_{0}, s_{1}, s_2,\dots)$ based on each initial state $\underline{\s}_0$ and its  feedback function $f$. The shift register sequence can be equivalently expressed as a sequence of states, $(\underline{\s}_0,\underline{\s}_1,\underline{\s}_2, \dots)$,  with the relation $\underline{\s}_i = \calF(\underline{\s}_{i-1})= \cdots = \calF^i(\underline{\s}_0)$ for $i\geq 0$.
	When $\underline{\s}_p =\calF^p(\underline{\s}_0)= \underline{\s}_0$ for the least integer $p\geq 1$, we obtain a cycle of states $\underline{\s}_0,\dots \underline{\s}_{p-1}$, or equivalently a sequence $(s_0,\dots, s_{p-1}, \dots)$ of period $p$.
	
	\begin{figure}[!h]
		\begin{center}
			\setlength{\unitlength}{1mm}
			\begin{picture}(110,35)
				\put(15,25){\framebox(12.3,7)[c]{$s_i$}}\put(33.1,25){\framebox(12.3,7)[c]{$s_{i+1}$}}
				\put(63.8,25){\framebox(12.3,7)[c]{$s_{i+m-2}$}}\put(81.8,25){\framebox(12.3,7)[c]{$s_{i+m-1}$}}
				\put(51,25){\makebox(7,7)[c]{$\cdots$}}
				\put(33.1,28.5){\vector(-1,0){5.7}}\put(51.1,28.5){\vector(-1,0){5.7}}
				\put(63.8,28.5){\vector(-1,0){5.7}}\put(81.8,28.5){\vector(-1,0){5.7}}
				\put(15,28.5){\vector(-1,0){15}}
				\put(16,8){\framebox(78,6.8)[c]{feedback function $f$}} 
		\put(20.5,24.8){\vector(0,-1){10}}\put(38.5,24.8){\vector(0,-1){10}}
		\put(69.5,24.8){\vector(0,-1){10}}\put(87.5,24.8){\vector(0,-1){10}}
		\put(94,11.5){\line(1,0){10}}\put(104,11.5){\line(0,1){17}}\put(104,28.5){\vector(-1,0){10}}
	\end{picture}
\end{center}
\vspace{-1.2cm}\caption{An $m$-stage FSR with feedback function $f$}
\label{Fig1}
\end{figure}
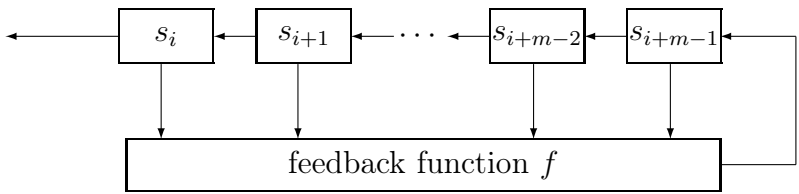

In his influential book \cite{Golomb2017}, Golomb intensively studied the relation between
the feedback function of an FSR and its output sequences/cycles.
He showed that an FSR generates disjoint cycles if and only if it uses nonsingular feedback functions of the form $f(x_0,x_1,\dots, x_{m-1}) = x_0 + g(x_1, \dots, x_{m-1})$.
Nonsingular FSRs are of practical interest as their output sequences have simpler structure and often exhibit desirable randomness properties \cite{Golomb2017}. For an $m$-stage nonsingular binary FSR,
when its feedback function $f$ is linear, the output sequence can have the longest period $2^m-1$, which contains
all nonzero $m$-tuples exactly once and is	known as a \textit{maximum-length} sequence (\textit{m}-sequence for short);  when $f$ is nonlinear, its output sequence can have the longest period $2^m$, containing all binary $m$-tuples exactly once, and is known as a \textit{binary de Bruijn sequence of order $m$} \cite{deBruijn}.
Both \textit{m}-sequences and de Bruijn sequences exhibit the \textit{span property} \cite{Golomb_1980_, Kim_2023_}. They are of significant interest in research and applications \cite{Fredricksen_1982_, Helleseth-Li}.
While the theory of \textit{m}-sequences is well explored, many problems about de Bruijn sequences remain unsolved, for instance,
the necessary or sufficient properties of feedback functions for generating de Bruijn sequences, more explicit combinatorial structure of de Bruijn sequences and efficient generations of all de Bruijn sequences of modest lengths, etc.

\subsection{Nonlinear complexity}

Linear complexity profile has been an important measure of randomness of sequences for cryptography \cite{Rueppel1986}.
As an additional figure of merit to judge the randomness of sequences, Jansen and Boekee proposed the
maximum-order complexity, later known as nonlinear complexity, of sequences \cite{Jansen, JansenPhD}.

\begin{defn} (\cite{JansenPhD}) \label{nlc}
The nonlinear complexity of a sequence $\textbf{s}$ over an alphabet $\mathcal{A}$, denoted by $nlc(\textbf{s})$, is the length of the shortest feedback shift registers that can generate the sequence $\textbf{s}$. \cn
\end{defn}

For a sequence $\s=(s_0,s_1,\dots)$ over $\mathcal{A}$, the term $s_{i+k}$ is deemed as
the \textit{successor} of the subsequence $\s_{[i:i+k]}=(s_i, \dots, s_{i+k-1})$ for certain
positive integers $i$ and $k$.
Some properties of the nonlinear complexity of sequences are recalled below.

\begin{lem}(\cite{JansenPhD})\label{lem}
The nonlinear complexity of a sequence $\textbf{s}$ equals one plus the length of its longest identical subsequences that occur at least twice with different successors.
\end{lem}

The metric of nonlinear complexity is well defined
for both finite-length sequences and infinite-length periodic sequences, which exhibit apparent difference in spite of the close connection.
Below we recall some basics for finite-length sequences and periodic sequences, respectively.

\begin{lem}(\cite{JansenPhD})\label{lem1}
For a finite-length sequence
$(s_0,s_1, \dots, s_{n-1})$ over an alphabet $\mathcal{A}$,

\noindent{\rm (i)} its nonlinear complexity takes values ranging from 0 to $n-1$;

\noindent{\rm (ii)}
if its nonlinear complexity $c \geq \frac{n}{2}$, then
the sequence $(s_{0}, s_1\dots, s_{n-1}, s_n)$ for any new symbol $s_n$ in $\mathcal{A}$ has the same nonlinear complexity $c$;

\noindent{\rm (iii)} if its nonlinear complexity $c\geq \n$, then
it cannot be written as $(s_0,s_1,\dots, s_{k-1})^e$ for any proper divisor $k$ of $n$.

\end{lem}

Throughout what follows, we will use $\sn$ to denote a sequence $(s_0, s_1,\dots, s_{n-1})$ that is not a repetition of any shorter subsequence and deem it
as an aperiodic finite-length sequence.
Given a sequence $\textbf{s}_{n}=(s_{0},\dots,s_{n-1})$, the left circular shift operators $L^i(\sn)$ on $\textbf{s}_n$ are defined as $L^0(\sn) = \sn$,
$L(\sn)=(s_1, \dots, s_{n-1}, s_0)$ and
\(L^{i}(\textbf{s}_{n})=L(L^{i-1}(\sn))=(s_{i},s_{i+1},\dots,s_{n-1},s_{0}\dots,s_{i-1}) \) for $i\geq 1$.
The right circular shift operators are defined by $R^{i}(\textbf{s}_{n})=L^{n-i}(\textbf{s}_n)$.
Under circular shift operators, we can derive from $\sn$ a \textit{shift equivalence class} $\{\textbf{s}_n, L(\textbf{s}_n), \dots, L^{n-1}(\textbf{s}_n)\}$.
Notice that the nonlinear complexity of $\s_n$ is usually not an invariant under circular shift operators.
For instance, while the sequence $\textbf{s}_n=(1,0,\dots,0)$ has nonlinear complexity $1$,  its shift sequence $L(\textbf{s}_n) = (0, \dots, 0, 1)$ has maximum nonlinear complexity $n-1$.

Shift operators can be similarly defined for periodic sequences.
The shift equivalence class of a periodic sequence $\sni$ is given by $\{\textbf{s}_n^{\infty}, (L(\textbf{s}_n))^{\infty}, \dots, (L^{n-1}(\textbf{s}_n))^{\infty}\}$.
As recalled in the following lemma, the nonlinear complexity of a periodic sequence is an invariant under circular shift operators.

\begin{lem}(\cite{JansenPhD})\label{per-nlc}
Let $\sni$ be a sequence of period $n$ over an alphabet $\mathcal{A}$ of size $q$. Then,

\noindent{\rm (i)}
the nonlinear complexity of $\textbf{s}_n^{\infty}$ satisfies $\lceil\log_{q}(n) \rceil \leq nlc(\sni)\leq  n-1$;

\noindent{\rm (ii)} all sequences in $\{\textbf{s}_n^{\infty}, (L(\textbf{s}_n))^{\infty}, \dots, (L^{n-1}(\textbf{s}_n))^{\infty}\}$
have the same nonlinear complexity.
\end{lem}

For finite-length sequences $\sn$ and periodic sequences $\sni$, despite their close connection, they
can behave differently in some aspects. As indicated in Lemma \ref{lem1} and Lemma \ref{per-nlc}, the nonlinear complexities of $\sn$ and $\sni$ take values from different ranges, and they
behave differently under circular shift operators.
For a finite-length sequence $\sn$, since its corresponding periodic sequence $\sni$ contains all shift equivalent sequences $\textbf{s}_n, L(\textbf{s}_n), \dots, L^{n-1}(\textbf{s}_n)$,
it is clear that $\sni$ has nonlinear complexity $nlc(\sni) \geq nlc(L^i(\sn))
$ for any $0\leq i <n$, i.e.,
$$nlc(\sni) \geq \max\limits_{0\leq i <n} nlc(L^i(\sn)).$$ While for some sequences $\sn$
the equality of the above inequality can be achieved, the equality can not be reached for many other sequences
$\sn$, i.e., $nlc(\sni) > \max\limits_{0\leq i <n} nlc(L^i(\sn)).$ For instance, for the sequence $\textbf{s}_{10}=(0010010010)$,
we have
$$
\left(\,nlc(L^i(\textbf{s}_{10}))\,:\, i=0,1,2,\dots, 9\,\right) = (\,2, 2, 7, 6, 5, 4, 6 ,5, 4 ,3\,)\, \text{ and }\, nlc(\s_{10}^{\infty}) = 9.
$$

In this paper we will further investigate the varying behaviour of nonlinear complexity of binary sequences $\sn$ under circular shift operators, thereby
characterizing
periodic binary sequences $\sni$ with a prescribed nonlinear complexity.
As a result, we are enabled to reveal more explicit relations between $nlc(\sni)$ and $nlc(\mathbf{a}_n)$, where $\mathbf{a}_n$ is a certain $n$-length subsequence of $\sni$.
Throughout this paper,
we will denote
$c=nlc(\sn)$ and $\omega = nlc(\sni)$ for a finite-length sequence $\sn$  and the corresponding periodic sequence $\sni$,
respectively.

Suppose $\textbf{a}_n^{\infty}$ is a binary sequence with period $n$ and nonlinear complexity $\omega$ generated by an $\omega$-stage FSR.
Recall from Figure \ref{Fig1} that $\textbf{a}_n^{\infty}=(a_0,a_1,a_2,\dots)$ can be equivalently expressed as a cycle of states $(\underline{\textbf{a}}_{0},\underline{\textbf{a}}_{1},\dots,\underline{\textbf{a}}_{n-1})$, where $\underline{\textbf{a}}_t = \calF^t(\underline{\textbf{a}}_0)$ for $t \geq 0$.
Given a state $\underline{\textbf{a}}_i$, its \textit{companion state}
is defined by $\widehat{{\underline{\textbf{a}}}_{i}}= (a_{i},\dots,a_{i+\omega-2}, \overline{a}_{i+\omega-1})$
with $\overline{a}_{i+\omega-1}={a}_{i+\omega-1}\oplus 1$.
A companion pair $(\underline{\textbf{a}}_i, \widehat{{\underline{\textbf{a}}}_{i}})=(\underline{\textbf{a}}_i,\underline{\textbf{a}}_{i+d})$ for certain  $d\geq 1$  can be denoted as $ (\underline{\textbf{a}}_{i},\calF^{d}({\underline{\textbf{a}}}_{i}))$.
The structure of $(\underline{\textbf{a}}_{i},\calF^{d}({\underline{\textbf{a}}}_{i}))$ will be frequently used in our discussion.
Without loss of generality, we assume
$d\leq \left\lfloor\frac{n}{2}\right\rfloor$ (since for the case that $d> \n$ we can consider the companion pair $(\underline{\textbf{a}}_{i+d}, \calF^{n-d}(\underline{\textbf{a}}_{i+d})))$.
We shall consider its left shift sequence $\textbf{s}_n^{\infty}=(L^i(\textbf{a}_n))^{\infty}$, which has a companion pair $(\underline{\textbf{s}}_{0}, \calF^{d}({\underline{\textbf{s}}}_{0}))$ with $d\leq \n$.
Here we consider $\sni$ instead of $\textbf{a}_n^\infty$ for simplifying the notation in subsequent sections.

\section{Characterizations of finite-length sequences}\label{Sec2.2}

In this section, we consider finite-length sequences with certain structure and discuss the change of nonlinear complexities of those sequences under circular shift operators.

We first recall recent results on the nonlinear complexity of binary sequences.

\begin{lem}(\cite{liang})\label{lem unique}
For	a binary finite-length sequence
$(s_0,s_1, \dots, s_{n-1})$,
if it has nonlinear complexity  $c \geq \frac{n}{2}$, then
there exists exactly one pair of identical subsequences of length $c -1$ with different successors in $(s_{0}, s_1, \dots, s_n)$.
\end{lem}

\begin{lem}(\cite{liang})\label{aper-puanduan}
Let $c$ and $d$ be two positive integers with $c+d \geq 3$, and
a  binary sequence given by
\begin{equation}\label{c+d}
	\textbf{s}_{c+d}=\obr{(s_0, \dots, s_{d-1})}{q \text{ repetitions}}(s_0, \dots, s_{r-1},\overline{s}_{r})=
	\s_d^{q} \,(s_0, \dots, s_{r-1},\overline{s}_{r})
\end{equation}
where $q = \lfloor \frac{c+d-1}{d} \rfloor$,  $0 \leq r=(c+d-1)-qd<d$, $\s_d$ is aperiodic and $\s_d^q=\obr{\s_d}{q \text{ repetitions}}$,
and when $r=0$, $(s_0, \dots, s_{r-1},\overline{s}_{r})$ is deemed as $ \overline{s}_{0}$.
Then when $c \geq d$, the sequence $\s_{c+d}$ in \eqref{c+d} has nonlinear complexity $c$.	
\end{lem}

The binary sequences of the form in \eqref{c+d} will be heavily used in subsequent discussions. 
Note that $\s_{c+d}$ can be equivalently expressed as
$$(s_0,s_1,\dots, s_{c-2})=(s_d,s_{d+1},\dots, s_{c+d-2}) \text{ and } s_{c-1} \neq s_{c+d-1}.$$
Given a sequence $\textbf{s}_{c+d}=(s_0, s_{1}, \dots s_{c+d-1})$, its two subsequences $\s_{[0:c]}$ and $\s_{[d:c+d]}$
are overlapped at components $s_{d}, \dots, s_{c-1}$ when $c>d$ and are exactly next to each other when $c=d$.
Liang et al. \cite{liang} showed that when $c\geq d$ the sequence $\s_{c+d}$ has nonlinear complexity $c$ and spacing $d$ between its companion pair
$(\underline{\textbf{s}}_{0},\calF^{d}({\underline{\textbf{s}}}_{0}))$
if and only if it has the form as in (\ref{c+d}).
When $c<d$, the statement is neither true for sufficiency nor for necessity. 	
The result in Lemma \ref{aper-puanduan} helps us calculate the nonlinear complexity of a binary sequence $\sn$ starting with $\s_{c+d}$ in a direct way.

\begin{cor}\label{B_nlc}
Suppose a binary sequence $\sn=(s_0,\dots,s_{n-1})$ has
its subsequence $\textbf{s}_{c+d}$ of the form in \eqref{c+d} for certain positive integers $c$ and $d$ with $1\leq d\leq \min\{n-c,\n\}$.
Then we have $nlc(\sn)\geq c$, where the equality is achieved when $c\geq \n$.
\end{cor}

\begin{proof}
When $c \geq d$, it follows from Lemma \ref{aper-puanduan} that $nlc(\textbf{s}_{c+d})=c$.
Then $nlc(\sn) \geq nlc(\textbf{s}_{c+d})=c$.
When $c < d$, the subsequence $\s_{c+d}$ is of the form
\[
\s_{c+d} = (s_0,\dots, s_{c-2}, s_{c-1},\dots, s_{d-1})(s_0,\dots, s_{c-2},\overline{s}_{c-1}),
\] where the subsequences $(s_0,\dots, s_{c-2})$ have different successors. By Lemma \ref{lem} we have
$nlc(\textbf{s}_{c+d})\geq c$.
So it is clear that $nlc(\sn)\geq c$.

If $c\geq \n$, then by Lemma \ref{lem1} (ii) recursively, we have $nlc(\sn) =nlc(\textbf{s}_{c+d})$.
Note that $nlc(\textbf{s}_{c+d})=c$ since $c\geq \n\geq d$.
Thus $nlc(\sn) =c$ if $c\geq \n$.
\end{proof}

Below we define a set $\mathcal{B}(n,c)$ that
helps us establish a connection between nonlinear complexities of $\sni$ and $\sn$.

\begin{defn}\label{BB}
For $1 \leq c< n$ and $1\leq d\leq \min\{n-c,\n\}$, we denote by $\mathcal{B}(n,c,d)$ the set of aperiodic binary sequences $\sn$ starting with $\s_{c+d}$ in \eqref{c+d}, namely,
\begin{equation}\label{structure of finite}
	\mathcal{B}(n,c,d)=\{\,\textbf{s}_{n}=\textbf{s}_{c+d}\,\textbf{s}_{[c+d:n]}
	=\underbrace{((s_0, \dots, s_{d-1})^{q} \, (s_{0},\dots, s_{r-1}, \overline{s}_{r})}_{\text{length}=c+d}\, \textbf{s}_{[c+d:n]})\, \}
\end{equation}
where $\s_d$ is aperiodic, $\overline{s}_{r}=s_{r}\oplus 1 $, and the subsequence $\textbf{s}_{[c+d:n]}$
is chosen from $\mathbb{Z}_2^{n-c-d}$ such that $\textbf{s}_{n}$ is aperiodic.
We call the parameter $d$ the spacing of $\sn$, denoted by $spac(\textbf{s}_{n})$, and
define
\begin{equation}\label{Eq-Bnc}
	\mathcal{B}(n,c)=\bigcup\limits^{\min\{n-c,\n\}}_{d=1}\mathcal{B}(n,c,d).
\end{equation}
\end{defn}

By Corollary \ref{B_nlc},  for the case of $c\geq \n$, any $\sn\in \mathcal{B}(n,c)$ has nonlinear complexity $c$.
When $c\geq \frac{n}{2}$, Lemma \ref{lem unique}  shows that $\s_{c+d}$ actually contains the unique companion pair in $\sn$; moreover, Lemma \ref{lem1} (ii) indicates that the subsequence $\textbf{s}_{[c+d:n]}$ in $\sn$ can be arbitrarily chosen.
It is to be noted that for each given $n,c$ and $d$, the set 
$\mathcal{B}(n, c,d)$ is non-empty.
As a matter of fact, as the subsequence $\s_{c+d}$ is aperiodic and $\s_{[c+d:n]}$ can be arbitrarily chosen. 
The existence of aperiodic sequences in $\mathcal{B}(n, c,d)$ can be easily confirmed.

The following lemma characterizes the change of nonlinear complexities of sequences in $\mathcal{B}(n,c)$ under the left circular shift operators.
\begin{lem}\label{c-c+1}
For $\sn\in\mathcal{B}(n,c, d)$ and a positive integer $t< c$, its left shift sequence
$L^t(\sn)$ belongs to $\mathcal{B}(n,c-t,d)$.
In particular, when $c-t \geq \n$, we have
$$nlc(L^t(\sn))=nlc(\sn)-t=c-t.$$
\end{lem}
\begin{proof}
According to Definition \ref{BB}, the sequence $\sn \in\mathcal{B}(n,c, d)$ has the form
\begin{equation}\label{a_n-c}	
	\textbf{s}_{n}=\s_{c+d}\s_{[c+d:n]}=\underbrace{(s_0, \dots, s_{d-1})^{q}(s_0, \dots, s_{r-1}, \overline{s}_{r})}_{\text{length}=c+d} \, \s_{[c+d:n]},
\end{equation}
which implies $(s_0,s_1,\dots, s_{c-2})=(s_d,s_{d+1},\dots, s_{c+d-2})$ and $s_{c-1} \neq s_{c+d-1}$.
Let $\textbf{a}_n=L(\s_n) $. Then
$(a_0,\dots, a_{c-3})=(s_1,\dots, s_{c-2})=(s_{d+1},\dots, s_{c+d-2})=(a_d,\dots, a_{c+d-3})$ and $a_{c-2} =s_{c-1}\neq s_{c+d-1} =a_{c+d-2}$.
Thus we can write $\textbf{a}_n$ as
\begin{equation}\label{a_n-c-1}	
	L(\s_n)=\textbf{a}_n= \textbf{a}_{c+d-1} \textbf{a}_{[c+d-1:n]}
	= \underbrace{(a_0, \dots, a_{d-1})^{q_1}(a_0,\dots, a_{r_1-1}, \overline{a}_{r_1})}_{\text{length} =c+d-1} \,\textbf{a}_{[c+d-1:n]},
\end{equation}	
where $\textbf{a}_d=(a_0,a_1,\dots, a_{d-1}) =(s_{1},s_{2},\dots,s_{d-1},s_{0})$,
$q_1=\lfloor \frac{(c+d-1)-1}{d} \rfloor=\lfloor \frac{c+d-2}{d} \rfloor$, $r_1=(c+d-2)-q_1d$ and $(a_0,\dots, a_{r_1-1}, \overline{a}_{r_1})$ reduces to $\overline{a}_{r_1}$ when $r_1=0$.
It is clear that $\textbf{a}_d = L(\s_d)$ is aperiodic.
Thus we have $L(\s_n) \in \mathcal{B}(n,c-1,d)$. Furthermore, it follows from Corollary \ref{B_nlc} that $L(\s_n) $ has nonlinear complexity $c-1$ when $c-1\geq \n$.
By repeatedly applying the above process, the desired statement follows.
\end{proof}

Lemma \ref{c-c+1} can be interpreted alternatively: given a sequence $\textbf{a}_n \in \mathcal{B}(n,c-1,d)$ in \eqref{a_n-c-1}
with $d \leq n-c$ and $c-1\geq \n$, if $a_{n-1}=a_{d-1}$,
then its right cyclic shift sequence $R(\textbf{a}_n)$ has the form in \eqref{a_n-c},
implying that $\sn = R(\textbf{a}_n)$ belongs to $\mathcal{B}(n,c,d)$ and has nonlinear complexity $nlc(\textbf{a}_n)+1$.
With this observation, we introduce the following parameter of a sequence in $\mathcal{B}(n,c, d)$ with $c\geq \lfloor\frac{n}{2}\rfloor$, which indicates
the potential increment of the nonlinear complexity of sequences under the right circular shift operators.

\begin{defn}\label{def-t}
Given a certain positive integer $t$,
if a sequence $\textbf{s}_{n}$ in $\mathcal{B}(n,c, d)$ satisfies 
\begin{equation*}
	s_{n-1-i}=s_{(d-1-i)\,\text{mod}\,\,d}\,\, \text{ for } 0\leq i< t, \text{ and } s_{n-1-t}\neq  s_{(d-1-t)\,\text{mod}\,\,d},
\end{equation*}
then we call $s_{n-t}, \dots, s_{n-1}$ the added terms of $\textbf{s}_{n}$
and denote by $add(\textbf{s}_{n})$ the number $t$ of the added terms of $\sn$.
\end{defn}

Definition \ref{def-t} indicates that
$(s_{n-t},\dots, s_{n-1})=(s_{(d-t)\,\text{mod}\,\,d},\dots, s_{(d-1)\,\text{mod}\,\,d})$.
As shown in Figure \ref{fig.2}, for the case that $t< d$, the subsequences $\s_{[d-t:d]}$ and $\s_{[n-t:n]}$ in gray are identical and  $s_{d-1-t}\neq s_{n-1-t}$.
In this case $s_{n-t}, \dots, s_{n-1}$ in gray are the added terms of $\sn$. For some sequences $\sn$, we may have $add(\sn)=t\geq d$.
In such a case, it follows that $(s_{n-t}, \dots, s_{n-d-1}, s_{n-d}, \dots, s_{n-1})=(s_{n-(t-d)},\dots, s_{n-1},s_0, \dots, s_{d-1})$.
\begin{figure}[H]
\centering
\begin{tikzpicture}	
	\node (rect) at (0,0)
	{
		$\begin{array}{rl}
			\sn &= (s_0, {\dots, s_{d-1}})^{q} (s_0, \dots, s_{r-1}, \overline{s}_{r}) \,(s_{c+d}, {\dots,s_{n-1}})
			\\
			&= (s_0, {\dots, s_{d-1-t}, s_{d-t}, \dots, s_{d-1}})(s_{d}, \dots, s_{n-t-1},{s_{n-t},\dots,s_{n-1}}).
		\end{array} $};
	\fill[gray, opacity=0.3] (-2,-0.6) rectangle (0.4,-0.15) ;
	\fill[gray, opacity=0.3] (3.3,-0.6) rectangle (5.7,-0.15) ;
	\draw[<-] (3.3,-0.7) -- (5.7,-0.7) ;
	\draw[<-] (-2,-0.7) -- (0.4,-0.7) ;
\end{tikzpicture}
\caption{Added terms in  Definition \ref{def-t} for $t< d$ }\label{fig.2}
\end{figure}

\begin{remark}\label{mark1.1}
Note that each $\textbf{s}_{n}$ in $\mathcal{B}(n,c)$ with $c\geq \lfloor\frac{n}{2}\rfloor$ has $add(\sn)\leq n-c-1$.
For $\textbf{s}_{n}$ in $\mathcal{B}(n,c,d)$ with $c\geq \lfloor\frac{n}{2}\rfloor$, assume $add(\textbf{s}_{n})=t$,
we now consider the sequence $\sn^2=\sn\sn$.
From Definition \ref{def-t}, in the subsequence
$(s_{n-t},\dots,s_{n-1},s_{0},\dots,s_{c+d-1}) $
of $\sn^2$ we have
\begin{equation*}
	(s_{n-t},\dots,s_{n-1},s_{0},\dots,s_{c-2})=(s_{n-t+d},\dots,s_{d-1},s_{d},\dots,s_{c+d-2}) \ \  \mbox{and} \  \ s_{c-1}\neq s_{c+d-1}.
\end{equation*} Let $\textbf{v}_d=(v_0, \dots, v_{d-1}) = (s_{n-t}, \dots, s_{n-t+d-1})$. Then
\begin{equation*}
	(s_{n-t},\dots,s_{n-1},s_{0},\dots,s_{c+d-1}) = \overbrace{(v_0, \dots, v_{d-1})\cdots (v_0, \dots, v_{d-1})}^{q_2 \text{ repetitions}}(v_0, \dots, v_{r_2-1}, \bar{v}_{r_2}),
\end{equation*} where $q_2 = \lfloor\frac{c+d+t-1}{d}\rfloor$ and  $r_2 = (c+d+t-1) - q_2d.$
Thus
the above subsequence satisfies
$(s_{n-t},\dots,s_{n-1},s_{0},\dots,s_{c+d-1}) \in\mathcal{B}(c+d+t,c+t,d)$ with nonlinear complexity $c+t$.
Then $c+t\leq nlc(\sn^2)\leq n-1$, implying $t\leq n-c-1$.

\end{remark}

For a binary sequence $\sn\in\mathcal{B}(n,c)$,  the following proposition shows that
$add(\sn)$ plays an important role in the varying behaviour of the nonlinear complexity of sequences under the right circular shift operators.

\begin{prop}\label{fini-shift-nlc}
For $\textbf{s}_{n}\in  \mathcal{B}(n,c,d)$ with $c \geq \n$, $d \leq \min\{n-c,\n\}$ and $add(\textbf{s}_{n})=t$,
the nonlinear complexity values of its shifted sequences have the following properties,

\noindent{\rm (i)} for any $1\leq k\leq \min\{t, n-c-d\}$,
$R^{k}(\textbf{s}_{n}) \in \mathcal{B}(n,c+k,d)$ and
$ nlc(R^{k}(\textbf{s}_{n}))=c+k;$

\noindent{\rm (ii)} for any $t< k\leq  n-c-d$,
we have $ nlc(R^{k}(\textbf{s}_{n}))=c+t.$
\end{prop}
\begin{proof}
(i) For $\textbf{s}_{n}\in  \mathcal{B}(n,c,d)$ with $add(\textbf{s}_{n})=t$,
we have
$(s_{n-t},\dots, s_{n-1})=(s_{(d-t)\,\text{mod}\,\,d},\dots,$\\
$ s_{(d-1)\,\text{mod}\,\,d})$, thus
$$(s_{n-t},\dots, s_{n-1})(s_0,\dots, s_{c-2})=(s_{(d-t)\,\text{mod}\,\,d},\dots, s_{(d-1)\,\text{mod}\,\,d})(s_d,s_{d+1},\dots, s_{c+d-2}),$$
and $s_{c-1} \neq s_{c+d-1}.$
When $ n-c-d \geq 1$,
let $\textbf{a}_n=R(\s_n) $,
thus
$$\begin{array}{rl}
	(a_{0},a_1,\dots, a_{c-1})&= (s_{n-1},s_0,\dots, s_{c-2})
	=(s_{d-1},s_d,\dots, s_{c+d-2})
	=	(a_{d},a_{d+1},\dots, a_{c+d-1})
\end{array}$$
and $a_{c} \neq a_{c+d}.$
That is to say, we can write $\textbf{a}_n$ as
\begin{equation*}
	R(\s_n)=\textbf{a}_n=  \textbf{a}_{c+d+1} \textbf{a}_{[c+d+1:n]}
	= \underbrace{(a_0, \dots, a_{d-1})^{q_2}(a_0,\dots, a_{r_2-1}, \overline{a}_{r_2})}_{\text{length}=c+d+1} \,\textbf{a}_{[c+d+1:n]},
\end{equation*}	
where
$q_2=\lfloor \frac{(c+d+1)-1}{d} \rfloor = \lfloor \frac{c+d}{d} \rfloor $ and $r_2=(c+d+1)-1-q_2d$. So we see that
the sequence $\textbf{a}_n=R(\s_n)$ belongs to $\mathcal{B}(n,c+1,d)$ and has nonlinear complexity $c+1$.

When $t \leq n-c-d$,  by induction on $t$, we see that
the sequence $\textbf{a}_n=R^k(\s_n)$ with $k \leq t$ belongs to $\mathcal{B}(n,c+k,d)$ and has nonlinear complexity $c+k$.	
When $t > n-c-d$ and $k \leq n-c-d$, the statement holds similarly.
However, when $t > n-c-d$ and $n-c-d < k \leq t$, the sequence $\textbf{a}_n=R^k(\s_n)$ satisfies $a_{i}=a_{i+d}$ with $0\leq i< n-d$ and $d\leq  \min\{n-c,\n \} \leq\lfloor\frac{n}{2}\rfloor$, implying that $\textbf{a}_n$ is contained in the periodic sequence $\textbf{a}_d^{\infty}$ and then $nlc(\textbf{a}_n)\leq nlc(\textbf{a}_d^{\infty})\leq d-1\leq\lfloor\frac{n}{2}\rfloor-1$, that is to say,
$\textbf{a}_n=R^k(\s_n)$ no longer belongs to $\mathcal{B}(n,c+k,d)$.
The difference between cases of $t \leq n-c-d$ and $t > n-c-d$ are visualized in Figure \ref{fig 4}, where
a copy of $\sn$ is added on its left side and the gray area covers the sequence $\s_{c+d}$ as in \eqref{c+d}.
Therefore,  for any $1\leq k\leq \min\{t, n-c-d\}$,
$ nlc(R^{k}(\textbf{s}_{n}))=nlc(\sn)+k=c+k.$

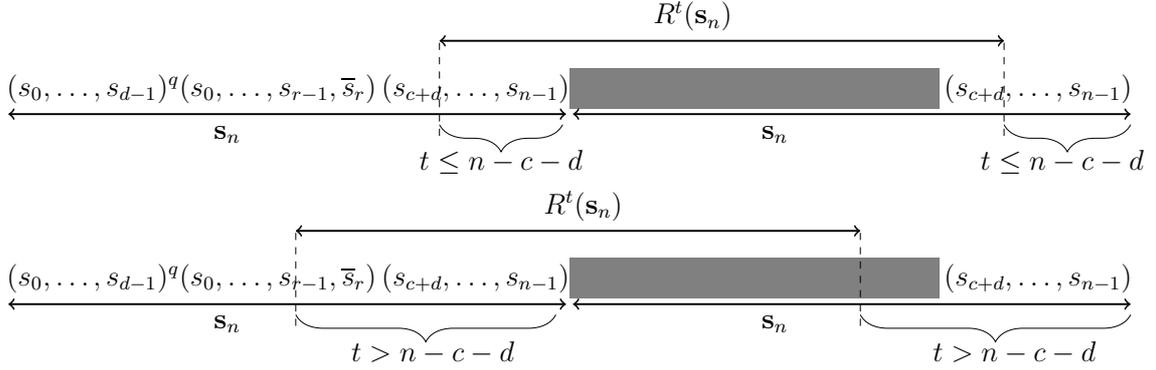
\begin{figure}[H]
	\begin{center}
		
		\scalebox{0.9}{
			\begin{tikzpicture}
				\node (rect) at (0,0)
				{$ (s_0, {\dots, s_{d-1}})^{q} (s_0, \dots, s_{r-1}, \overline{s}_{r}) \,(s_{c+d}, {\dots,s_{n-1}})    (s_0, {\dots, s_{d-1}})^{q} (s_0, \dots, s_{r-1}, \overline{s}_{r}) \,(s_{c+d}, {\dots,s_{n-1}})  $};
				\fill[gray, opacity=0.3] (0,-0.3) rectangle (5.4,0.3) ;
				\draw[thick, <->](0.05,-0.38)--(8.2,-0.38) node[below] at (3,-0.4) {$\sn$};
				\draw[<->,thick](-8.2,-0.38)--(-0.05,-0.38) node[below] at (-5,-0.4) {$\sn$};
				\draw[decorate,decoration={brace,amplitude=4mm,mirror,raise=1pt}] (-1.9,-0.5) -- (-0.1,-0.5) node[above] at (-1,-1.4) {$t\leq n-c-d$};
				\draw[<->,thick](-1.9,0.7)--(6.35,0.7) node[above] at (1.8,0.7) {$R^t(\sn)$};
				\draw[dashed] (-1.9,-0.7)--(-1.9,0.7);
				\draw[dashed] (6.35,-0.7)--(6.35,0.7);
				\draw[decorate,decoration={brace,amplitude=4mm,mirror,raise=1pt}] (6.35,-0.5) -- (8.2,-0.5) node[above] at (7.2,-1.4) {$t\leq n-c-d$};
				
				\begin{scope}[shift={(0,2.2)}]									
					\node (rect) at (0,-5)
					{$ (s_0, {\dots, s_{d-1}})^{q} (s_0, \dots, s_{r-1}, \overline{s}_{r}) \,(s_{c+d}, {\dots,s_{n-1}} )   (s_0, {\dots, s_{d-1}})^{q} (s_0, \dots, s_{r-1}, \overline{s}_{r}) \,(s_{c+d}, {\dots,s_{n-1}})  $};
					\fill[gray, opacity=0.3] (0,-5.3) rectangle (5.4,-4.7) ;
					\draw[thick, <->](0.05,-5.38)--(8.2,-5.38) node[below] at (3,-5.4) {$\sn$};
					\draw[<->,thick](-8.2,-5.38)--(-0.05,-5.38) node[below] at (-5,-5.4) {$\sn$};
					\draw[decorate,decoration={brace,amplitude=4mm,mirror,raise=1pt}]	 (-4,-5.5) -- (-0.2,-5.5) node[above] at (-2,-6.4) {$t>n-c-d$};
					\draw[<->,thick](-4,-4.3)--(4.25,-4.3) node[above] at (0.2,-4.3) {$R^t(\sn)$};
					\draw[dashed] (-4,-5.7)--(-4,-4.3);
					\draw[dashed] (4.25,-5.7)--(4.25,-4.3);
					\draw[decorate,decoration={brace,amplitude=4mm,mirror,raise=1pt}] (4.25,-5.5) -- (8.2,-5.5) node[above] at (6.5,-6.4) {$t>n-c-d$};	
				\end{scope}
			\end{tikzpicture}
		}
		
		\caption{The visualized description of Proposition \ref{fini-shift-nlc} (i)}\label{fig 4}
	\end{center}
\end{figure}

(ii)
Let $R^{t}(\textbf{s}_{n})=\textbf{a}_n$.
It follows from Proposition \ref{fini-shift-nlc} (i) and $t<n-c-d$ that
\begin{equation}\label{d-d'}
	R^{t}(\textbf{s}_{n})=\textbf{a}_n \in \mathcal{B}(n,c+t,d) \,\,\, \text{with} \,\,\, a_{n-1} \neq a_{d-1}	
\end{equation}
and
$ nlc(\textbf{a}_{n})=c+t.$
Then  $\textbf{a}_{c+t+d}$
is the subsequence formed exactly by the companion pair
$(\underline{\textbf{a}}_{0},\mathcal{A}^{d}({\underline{\textbf{a}}}_{0}))$  of $\textbf{a}_n$.
For any $t< k\leq  n-c-d$,
we have
$R^{k}(\textbf{s}_{n})=R^{k-t}(\textbf{a}_{n})$ contains $\textbf{a}_{c+t+d}$, which implies that $ nlc(R^{k}(\textbf{s}_{n})) \geq c+t.$
We shall prove the statement by induction on $k$.
Suppose that the nonlinear complexity of $R^{t+1}(\textbf{s}_{n})=R(\textbf{a}_{n})$ is larger than $c+t$, that is to say, $nlc(R(\textbf{a}_{n}))=c'>c+t$.
Note that it follows from Lemma \ref{lem unique} and $c' >\n $ that $R^{k}(\textbf{s}_{n})$ only has a unique companion pair.
If the unique companion pair of $R(\textbf{a}_{n})=(a_{n-1},a_0,\dots,a_{n-2})$ does not begin with $a_{n-1}$, then
$L(R(\textbf{a}_{n}))=\textbf{a}_n=(a_0,\dots,a_{n-2},a_{n-1})$ contains the unique companion pair of $R(\textbf{a}_{n})$, which implies that
$nlc(\textbf{a}_{n}) \geq nlc(R(\textbf{a}_{n}))=c'$.
It contradicts that  $nlc(\textbf{a}_{n})=c+t<c'=nlc(R(\textbf{a}_{n}))$.
Thus the unique companion pair of $R(\textbf{a}_{n})=(a_{n-1},a_0,\dots,a_{n-2})$ begins with $a_{n-1}$,
which yields that $R(\textbf{a}_{n})\in \mathcal{B}(n,c',d')$.
Hence $a_{n-1}=a_{d'-1}$.
Together with \eqref{d-d'}, we have $a_{d'-1}=a_{n-1}\neq a_{d-1}  $, implying that
$d' \neq d$.
According to Lemma \ref{c-c+1}, $L(R(\textbf{a}_{n}))=\textbf{a}_n\in \mathcal{B}(n,c'-1,d')$ with $nlc(\textbf{a}_n)=c'-1$ and $d' \neq d$.
Moreover $\textbf{a}_n\in \mathcal{B}(n,c+t,d)$ with $nlc(\textbf{a}_n)=c+t$,
thus $c'=c+t+1$. That is to say, $\textbf{a}_n$ satisfies $\textbf{a}_n\in \mathcal{B}(n,c+t,d)$ and $\textbf{a}_n\in \mathcal{B}(n,c+t,d')$ with $d' \neq d$.
From the structure of $\textbf{a}_n$ in \eqref{c+d}, it is clear that
\begin{equation}\label{prop_form}
	a_{i}=a_{i+d}, \, a_{i}=a_{i+d'}, \, 0 \leq i \leq c-2 \quad \text{and} \quad a_{c-1} \neq a_{c+d-1},\, a_{c-1} \neq a_{c+d'-1}.
\end{equation}

In the case of $c+t \geq n/2$, it follows from Lemma \ref{lem unique} that $\textbf{a}_n$ only has a unique companion pair, which
contradicts that $d' \neq d$.
In the case of $c=\n$ with odd $n$ and $t=0$, i.e. $c+t=\n$,
without loss of generality, suppose $d'<d$ then $1 \leq d'<d\leq \n$.
If $d < \n$,
then according to \eqref{prop_form}, one has
$a_{c-1-d}=a_{c-1}$ and $a_{c-1-d}=a_{c+d'-1-d}$ implying that $a_{c-1}=a_{c+d'-1-d}$,
while $a_{c+d'-1-d}=a_{c+d'-1}$ and $a_{c-1} \neq a_{c+d'-1}$ follows
$a_{c-1} \neq a_{c+d'-1-d}$. It is a contradiction.
If $d = \n$, then from \eqref{prop_form} we can see that
$$(a_c,a_{c+1},\dots,a_{c+d'-1})=(a_{c-d},a_{c+1-d},\dots,a_{c+d'-1-d})=(a_0,a_{1},\dots,a_{d'-1})=\textbf{a}_{d'},$$
thus its Hamming weight $\text{wt}(a_c,a_{c+1},\dots,a_{c+d'-1})=\text{wt}(\textbf{a}_{d'})$.
While again by \eqref{prop_form}, one has
$$(a_c,a_{c+1},\dots,a_{c+d'-2},a_{c+d'-1})=(a_{c\,\text{mod} \,\,d'},\dots,a_{(c+d'-2)\,\text{mod} \,\,d'},a_{(c+d'-1)\,\text{mod} \,\,d'} \oplus 1),
$$ where $(a_{c\,\text{mod} \,\,d'},\dots,a_{(c+d'-2)\,\text{mod} \,\,d'},a_{(c+d'-1)\,\text{mod} \,\,d'})$ is a shifted version of $\textbf{a}_{d'}$.
Hence
$\text{wt}(\textbf{a}_{d'})$
$ \neq \text{wt}(a_c,a_{c+1},\dots,a_{c+d'-1})$,
which is a contradiction.
Thus $$nlc(R(\textbf{a}_{n}))=nlc(R^{t+1}(\textbf{s}_{n}))=c'=c+t.$$
By induction on $k$ ranging from $t+1$ to $n-c-d$, the proof follows.
\end{proof}

Given a sequence $\sn\in \mathcal{B}(n,c,d)$,
Proposition \ref{fini-shift-nlc} reveals the changing pattern of $nlc(R^k(\sn))$ for $k\leq n-c-d$. However, when $k> n-c-d$,
the change of $ nlc(R^{k}(\sn))$ does not indicate a strong pattern.
Consider a sequence $\sn \in  \mathcal{B}(n,c)$ with $c\geq \n$. The varying behaviour of $nlc(R^k(\sn))$ for $0\leq k\leq n-1$ is discussed in Remark \ref{mark2}.

\begin{remark}\label{mark2} Without loss of generality, we assume $c=\n$. For $\sn \in  \mathcal{B}(n,\n)$,
among the right shift sequences $R^k(\sn)\,:\, k= 0, 1,\dots, n-1$, there may exist positive integers $0<  j_1 <\dots < j_r \leq n-1$ such that
$R^{j_1}(\sn), \dots, R^{j_r}(\sn)$ also belong to $\mathcal{B}(n,\n)$.

For $0\leq k < j_1$, assume $\sn \in \mathcal{B}(n,\n,d)$ for certain $d\leq \n$ has $add(\sn)=t$ and define an integer-valued
sequence $\textbf{c}_0=\left(nlc(\sn), \dots, nlc(R^{j_1-1}(\sn))\right)$.
Then the values in the sequence $\mathbf{c}_0$ vary in this way: $\textbf{c}_0$ starts with $\n$, when $k\leq \min\{t, n-c-d\}$, the values in $\textbf{c}_0$ increase one by one as in Proposition \ref{fini-shift-nlc} (i);
when $t< k\leq n-c-d$, the values in $\textbf{c}_0$ remain unchanged as in Proposition \ref{fini-shift-nlc} (ii); and when $n-c-d<k<j_1$,
the change of values in $\textbf{c}_0$ does not exhibit a strong pattern.

For $k=j_1$, assume $\textbf{a}_n=R^{j_1}(\sn)\in \mathcal{B}(n,c,d_1)$ has $add(\textbf{a}_n)=t_1$ and take
$$\textbf{c}_1=\left(nlc(R^{j_1}(\sn)), \dots, nlc(R^{j_2-1}(\sn))\right)=\left(nlc(\textbf{a}_n), \dots, nlc(R^{j_1'-1}(\textbf{a}_n))\right), $$  where $j_1'=j_2-j_1$. As $k$ increases from $j_1$ to $j_2-1$, the  values in the sequence $\textbf{c}_1$ vary in a similar manner as the above analysis for $\textbf{c}_0$.
As $k$ ranges from $j_2$ to $n-1$, we can similarly define the sequences $\textbf{c}_i$ for $i=2, \dots, r$ as above. The values in these sequences vary similarly to that of $\textbf{c}_0$.

To sum up, the above analysis shows that
\begin{equation*}
	\left(nlc(\sn), nlc(R(\sn), \dots, nlc(R^{n-1}(\sn))\right) =\textbf{c}_0 \cup \textbf{c}_1 \cup \cdots \cup \textbf{c}_r,
\end{equation*} where the values in each $\textbf{c}_i$, $i=0,1, \dots, r$, vary according to Proposition \ref{fini-shift-nlc} first and then change in unclear patterns.
\end{remark}

For sequences $\sn\in \mathcal{B}(n,c)$,  below we consider a subset $E(\sn)$ of its shift equivalence class and the representative of $E(\sn)$, which will be used in our subsequent discussions.

\begin{defn}\label{ES_Def}
Let $\sn$ be a sequence in $\mathcal{B}(n,c)$ and $E(\sn)=\{R^k(\sn): 0\leq k<n\} \cap \mathcal{B}(n,c)$.
A sequence $\widetilde{\s}_n \in E(\sn)$ satisfying
\begin{equation*}
	add(\widetilde{\s}_n) \geq add(\textbf{a}_n),\, \forall\, \textbf{a}_n\in E(\sn)
\end{equation*} is said to be a representative sequence of $\sn$. Furthermore, we define $\mathcal{R}(n,c)$ as  the set of all sequence representatives in $\mathcal{B}(n,c)$,
i.e.,
\begin{equation}\label{Eq-Rnc}
	\mathcal{R}(n,c)= \bigcup_{\sn\in \mathcal{B}(n,c)} \big\{\widetilde{\s}_n\in E(\sn):add(\widetilde{\s}_n) \geq add(\textbf{a}_n),   \forall\, \textbf{a}_n \in E(\sn) \big\}.
\end{equation}

\end{defn}

\begin{figure}[t!]
\centering	
\begin{tikzpicture}
	\begin{scope}[shift={(0,7)}]  		
		\datavisualization [school book axes,
		visualize as scatter,
		x axis={unit vector=(0:0.67pt), label={$k$}}, 
		y axis={length=4.04cm,label={$nlc(R^k(\sn))$}},
		y axis={grid={  at={4} , minor={  style=densely dashed,at={1, 2, 3,5,6} }}},
		]	
		data {
			x,y
			0,4
			1,5
			2,5
			3,5
			4,5
			5,4
			6,5
			7,6
			8,3	
		};
			\fill [black] (canvas cs:x=0cm,y=2.7cm) circle (2.5pt);
			\fill [black] (canvas cs:x=3.35cm,y=2.7cm) circle (2.5pt);
			\draw(3.2,2.55) rectangle (3.5,2.85) ;
			\draw[solid, thick, opacity=10](0,2.7) -- (0.67,3.37)--(1.34,3.37)--(2.01,3.37)--(2.68,3.37);		
			\draw[ thick,dashed, opacity=10](2.68,3.37)--(3.35,2.7);	
			\draw[solid, thick, opacity=10] (3.35,2.7)--(4.02,3.37)--(4.69,4.04);
			\draw[ thick,dashed, opacity=10](4.69,4.04)--(5.36,2.03);			
	\end{scope}	
	\node (rect) at (3,6){(i) }; \node at (3, 11.75) {$\sn=(\underline{\mathbf{0}0001}0010)$};

	\begin{scope}[shift={(9,7.05)}]
		\datavisualization [school book axes,
		visualize as scatter,
		x axis={unit vector=(0:0.67pt), label={$k$}},
		y axis={length=4.04cm,label={$nlc(R^k(\sn))$}},
		y axis={grid={  at={4} , minor={  style=densely dashed,at={1, 2, 3,5,6} }}},
		]
		data {
			x,y
			0,4
			0,6
			1,4
			2,4
			3,4
			4,5
			5,5
			6,4
			7,5
			8,3	
		};
		\fill [black] (canvas cs:x=0cm,y=2.7cm) circle (2.5pt);
		\fill [black] (canvas cs:x=2.01cm,y=2.7cm) circle (2.5pt);
		\fill [black] (canvas cs:x=4.02cm,y=2.7cm) circle (2.5pt);
		\draw(1.86,2.55) rectangle (2.16,2.85) ;
		\draw(3.87,2.55) rectangle (4.17,2.85) ;
		(0,4) -- (1,4)--(2,4)--(3,4)--(4,5)--(5,5)--(6,4)--(7,5)--(8,3)--(9,4);	
		\draw[solid,thick, opacity=10](0,2.7) -- (0.67,2.7)--(1.34,2.7)--(2.01,2.7);		
		\draw[thick,dashed, opacity=10](2.01,2.7)--(2.68 ,3.37)--(3.35,3.37)--(4.02,2.7);	
		\draw[solid,thick, opacity=10] (4.02,2.7)--( 4.69,3.37);
		\draw[thick,dashed, opacity=10]( 4.69,3.37)--(5.36,2.03);
	\end{scope}	
	\node (rect) at (12,6){(ii) };	\node at (12, 11.75) {$\sn=(\underline{\mathbf{10}1011}101)$};

	\begin{scope}[shift={(0,0)}]
		\datavisualization [school book axes,
		visualize as scatter,
		x axis={unit vector=(0:0.4pt), label={$k$}},
		y axis={length=4cm,label={$nlc(R^k(\sn))$}},
		y axis={grid={  at={7} , minor={  style=densely dashed,at={1, 2, 3,4,5,6,8,9,10} }}},
		]
		data {
			x,y
			0,7
			1,8
			2,9
			3,10
			4,10
			5,10
			6,10
			7,4
			8,4
			9,4
			10,4
			11,5
			12,5
			13,6
		};
		\scalebox{0.4}{
			\fill [black] (canvas cs:x=0cm,y=7cm) circle (5pt);	
			\draw(-0.3,6.7) rectangle (0.3,7.3) ;
			\draw[solid,very thick, opacity=10](0,7) -- (1,8)--(2,9)--(3,10)--(4,10)--(5,10)--(6,10);		
			\draw[very thick,dashed, opacity=10](6,10)--(7,4)--(8,4)--(9,4)--(10,4)--(11,5)--(12,5)--(13,6);
	}\end{scope}
	\node (rect) at (3,-1){(iii) };		\node at (3.5, 4.65) {$\sn=(\underline{\mathbf{0}0000001}111000)$};

	\begin{scope}[shift={(9,0)}]		
		\datavisualization [school book axes,
		visualize as scatter,
		x axis={unit vector=(0:0.4pt), label={$k$}},
		y axis={length=4cm,label={$nlc(R^k(\sn))$}},
		y axis={grid={  at={7} , minor={  style=densely dashed,at={1, 2, 3,4,5,6,8,9,10} }}},
		]
		data {
			x,y
			0,7
			1,8
			2,9
			3,10
			4,3
			5,3
			6,4
			7,5
			8,5
			9,6
			10,7
			11,8
			12,5
			13,6
		};
		\scalebox{0.4}{
			\fill [black] (canvas cs:x=13.5cm,y=7cm) circle (4pt);
			\fill [black] (canvas cs:x=23.5cm,y=7cm) circle (4pt);
			\draw(13.2,6.7) rectangle (13.8,7.3) ;
			\draw[solid,very thick, opacity=10](13.5,7) -- (14.5,8)--(15.5,9)--(16.5,10);		
			\draw[very thick,dashed, opacity=10](16.5,10)--(17.5,3)--(18.5,3)--(19.5,4)--(20.5,5)--(21.5,5)--(22.5,6)--(23.5,7);
			\draw[solid,very thick, opacity=10] (23.5,7)--(24.5,8);
			\draw[very thick,dashed, opacity=10](24.5,8)--(25.5,5)--(26.5,6);
		}	
	\end{scope}	
	\node (rect) at (12,-1){(iv) };	 \node at (12.5, 4.65) {$\sn=(\underline{\mathbf{0010}0010000}010)$};
\end{tikzpicture}
\caption{The visualized description of Example \ref{ex1}
}\label{fig 5}
\end{figure}

The following example illustrates some definitions and results in this section.

\begin{example}\label{ex1}
Consider $n=9$ and a binary finite-length sequence $\textbf{s}_{9}=(000010010)$.
It is clear that  $\textbf{s}_{9}=\s_5s_{[5:9]}=(\underline{\mathbf{0}0001}0010)\in \mathcal{B}(9,4,1)$ as in Definition \ref{BB},
where $\s_{c+d}$ is underlined and $\s_d$ is in bold for $d=1$ and $c=4$.  By Definition \ref{def-t} we have $add(\textbf{s}_{9})=1$ since $s_{8}=0=s_0$ (which is displayed in bold) and $s_7\neq s_0$.	
Since $add(\s_9)=t=1<n-c-d=4$, from Proposition \ref{fini-shift-nlc} we have $ nlc(R(\sn))=c+1=5$ and $ nlc(R^k(\s_9))=c+t=5$ for $k=2,3, 4$.

Furthermore, among the shift sequences $R^k(\s_9)$ for $k = 0, 1, \dots, 8$,
only the sequences $\s_9$ and $R^5(\s_9)$ belong to $\mathcal{B}(9,4)$, indicating that $E(\s_9) = \{\s_9,R^5(\s_9)\}$ from Definition \ref{ES_Def}.
Observe that
$$\left(nlc(R^k(\s_9)): k = 0, 1, \dots, 8\,\right)=(4,5,5,5,5)\cup (4,5,6,3\,) = \mathbf{c}_0 \cup \mathbf{c}_1$$
where the values in $\textbf{c}_0 =\left(nlc(R^k(\s_9)): k= 0, 1, 2,3, 4\,\right) $ and  $\mathbf{c}_1
=(nlc(R^k(\s_9)): k = 5, 6,7, 8)$ vary as in Proposition \ref{fini-shift-nlc} and the analysis is as 
in Remark \ref{mark2}.
In addition,   since $add(R^5(\s_9)) = 2>add(\textbf{s}_9)=1$, the representative of $E(\s_9) $ is $R^5(\s_9)$, which belongs to $\mathcal{R}(9,4)$ as in Definition \ref{ES_Def}.

The above discussion for the sequence $\s_9$ is visualized in Figure \ref{fig 5} (i),
where the axises represent the values $k$ and the corresponding $nlc(R^k(\s_9))$, respectively.
In Figure \ref{fig 5} (i), the sequences in $E(\s_9) $ are displayed as black solid dots,
and the representative sequence of $E(\s_9) $ is marked with a rectangle. In addition,
the solid line shows the varying behavior of $nlc(R^k(\s_9))$ that can be explained by Proposition \ref{fini-shift-nlc} and
the dotted line shows the varying behavior of $nlc(R^k(\s_9))$ that does not exhibit a strong pattern.

Sequences $\sn\in \mathcal{B}(n, \n)$ may behave quite differently. To illustrate this, we further consider  three
sequences
$
\s_9=(\underline{\mathbf{10}1011}101)\in \mathcal{B}(9,4,2),  \s_{14}=(\underline{\mathbf{0}0000001}111000)\in \mathcal{B}(14,7,1)$ and $\s_{14}=(\underline{\mathbf{0010}0010000}010)\in \mathcal{B}(14,7,4)
$
where $\s_{c+d}$ is underlined and $\s_d$ is in bold. For these sequences, Figure \ref{fig 5} (ii), (iii), (iv) displays the sequences in $E(\sn)$, the representative sequences in $E(\sn)$,
and the varying behavior of $nlc(R^k(\sn))$ for $k=0,1,\dots, n-1$, respectively.
\end{example}

At the end of this section,
we summarize important notations in Table \ref{tab_1} for readers' convenience.

\begin{table}[!h]
\caption{Some necessary notations of this paper.}
\label{tab_1}
\label{distab}
\renewcommand{\arraystretch}{1.2}
\begin{center}	
	\begin{tabular}{|c|c| p{1cm}|}
		\hline
		notation & description  \\
		\hline
		$\textbf{s}_{n}$  & an aperiodic sequence $(s_0, s_1, \dots, s_{n-1})$  \\
		
		\hline
		$\textbf{s}_{[i:i+k]}$  & the subsequence $(s_i,\dots,s_{i+k-1})$ \\
		\hline
		$\s_{c+d}$ & $\s_{d}^{q} \, (s_{0},\dots, s_{r-1}, \overline{s}_{r})$ as in \eqref{c+d} \\
		\hline
		$\mathcal{B}(n,c,d)$  & the set $\{\,\textbf{s}_{c+d}\,\textbf{s}_{[c+d:n]}
		\,\}
		$ with
		$1 \leq c<n$, $d\leq \min\{n-c,\n\}$ \\
		\hline
		$\mathcal{B}(n,c)$  & $\bigcup^{\min\{n-c,\n\}}_{d=1}\mathcal{B}(n,c,d)$   \\
		\hline
		$add(\textbf{s}_{n})$  &  for $\textbf{s}_{n}\in \mathcal{B}(n,c,d)$, the integer $t$ such that \\ &
		$s_{n-1-i}=s_{(d-1-i)\,\text{mod}\,\,d}, \forall \, 0\leq i< t, $ and  $s_{n-1-t} \neq  s_{(d-1-t)\,\text{mod}\,\,d}$ \\
		\hline
		$\mathcal{B}_0(n,c)$	&  $\mathcal{B}_0(n,c)=\{\sn \in \mathcal{B}(n,c):  add(\sn)=0\}$ \\
		\hline
		$E(\sn)$ & $\{R^k(\sn):\, 0\leq k<n \} \cap \mathcal{B}(n,c)$ for $\sn\in \mathcal{B}(n,c)$
		\\
		\hline
		$\mathcal{R}(n,c)$& $\bigcup_{\sn\in \mathcal{B}(n,c)} \left\{\widetilde{\s}_n\in E(\sn) :\,  add(\widetilde{\s}_n) \geq add(\textbf{a}_n),   \forall\, \textbf{a}_n \in E(\sn) \right\}$
		\\
		\hline \hline
		$\textbf{s}_n^{\infty}$  & the periodic sequence $(\textbf{s}_{n})^{\infty}$ from $\textbf{s}_{n}$  \\
		\hline
		$\mathcal{P}(n,\omega)$  & the set of sequences with period $n$ and nonlinear complexity $\omega$    \\ \hline
		$\overline{S}
		$
		&   $\overline{S} = \{L^k(\sni)\,:\, \sni\in S, 0\leq k<n\} $			\\
		\hline
		$S_1 \cong S_2$	& $\overline{S}_1=\overline{S}_2$ for two sets $S_1$ and $S_2$ of periodic sequences
		\\
		\hline			
	\end{tabular}
\end{center}
\end{table}

\section{The nonlinear complexity of $\sn$ and $\sni$}\label{Sec3}
This section investigates the relation between the nonlinear complexity of periodic sequences and that of finite-length sequences in $\mathcal{B}(n,c)$.
Let $\mathcal{P}(n,\omega)$ denote the set of $n$-periodic binary sequences with nonlinear complexity $\omega$, where  $\lceil\log_{2}(n)\rceil \leq \omega\leq n-1$ as indicated by Lemma \ref{per-nlc} (i).
Below we first discuss some properties of subsequences of a periodic binary sequence $\textbf{s}_n^{\infty}$.

\begin{lem}\label{lem_s_i^i+k+w}
Given a sequence $\textbf{a}_n^{\infty}$  in $\mathcal{P}(n,\omega)$, suppose its left shift sequence $\textbf{s}_n^{\infty}=L^i(\textbf{a}_n^{\infty})$ has a companion pair $(\underline{\textbf{s}}_{0}, \calF^{d}({\underline{\textbf{s}}}_{0}))$ with $d\leq \n$, where $\underline{\textbf{s}}_{0}=(s_0,\dots,s_{\omega-1})$. Then,

\noindent{\rm (i)} $\textbf{s}_{\omega+d}=(s_0,\dots,s_{\omega+d-1})\in \mathcal{B}(\omega+d,\omega,d)$ and $nlc(\textbf{s}_{\omega+d})= \omega$ if $\omega \geq d$;

\noindent{\rm (ii)} when $\omega+d\leq n$, one has
$\textbf{s}_n=(s_0,\dots,s_{n-1}) \in \mathcal{B}(n,\omega,d)$ and  $nlc(\textbf{s}_{n})= \omega$ if $\omega \geq d$;

\noindent{\rm (iii)} when $\omega+d> n$, for any $0\leq j< \omega-1$ and $n_1 =(\omega+d) - j$, the subsequence
$\s_{[j:\omega+d]}=(s_j,\dots,s_{\omega+d-1}) \in \mathcal{B}(n_1,n_1-d,d)$.
\end{lem}
\begin{proof}
{\rm (i)}
For the subsequence $\s_{\omega+d}$ that starts from $\underline{\textbf{s}}_0$ and ends with $\underline{\textbf{s}}_{d}=\calF^d(\underline{\textbf{s}}_0)=\widehat{\underline{\s}}_{0}$,
we have
$
(s_0, \dots, s_{\omega-2}) = (s_{d}, \dots, s_{\omega-2 + d }) $ and  $ s_{\omega-1}=s_{\omega-1+d}\oplus1.
$
Thus we can write $\s_{\omega+d}$ as
$$
\s_{\omega+d} =(s_0, \dots, s_{d-1}, s_{d}, \dots, s_{\omega+d-2}, \overline{s}_{\omega-1})= (\s_{d}^{q}\,(s_0, \dots, s_{r-1}, \overline{s}_{r})).
$$
The statement directly follows from
Definition \ref{BB}  and Lemma \ref{aper-puanduan}.

{\rm (ii)} If $\omega+d\leq n$, then $\textbf{s}_{\omega+d}$ is contained in $\textbf{s}_n$, namely, $$\textbf{s}_n=(s_0, \dots,s_{\omega+d-1},s_{\omega+d},\dots,s_{n-1})=\textbf{s}_{\omega+d}\,  \textbf{s}_{[\omega+d:n]}.$$
If $\omega \geq d$, then it implies $ nlc(\textbf{s}_n) \geq nlc(\textbf{s}_{\omega+d})= \omega$. Since $nlc(\textbf{s}_n)\leq nlc(\textbf{s}_n^{\infty})=\omega$, we have $nlc(\textbf{s}_n)=\omega$ and $\textbf{s}_n \in \mathcal{B}(n,\omega,d)$.

{\rm (iii)} For the case of $\omega+d> n$, since the subsequence $\textbf{s}_{\omega+d}$ belongs to $\mathcal{B}(\omega+d, \omega, d)$, it follows that
$s_{i_1} = s_{i_1+d} \text{ for } 0\leq i_1< \omega - 1 \text{ and } s_{\omega-1} = \overline{s}_{\omega-1 + d }.$
Hence for any $j \geq 0$, the subsequence $\s_{[j:\omega+d]}$ also satisfies the relation
\begin{equation}\label{Eq-Relation}
	s_{i_1} = s_{i_1+d} \ \text{ for }\  j\leq i_1 < \omega - 1 \ \text{ and } \ s_{\omega-1} = \overline{s}_{\omega-1 + d }.
\end{equation}
This implies that for any $j$ with $0\leq j< \omega-1$, the subsequence $\s_{[j:\omega+d]}$
has the form $$((s_j, \dots, s_{j+d-1})^{q}(s_j, \dots, s_{j+r-1}, \overline{s}_{j+r} ))$$ where $q = \lfloor \frac{\omega+d-j-1}{d}\rfloor$,
$r=(\omega+d-j-1)-qd$.
According to Definition \ref{BB},
the subsequence $\s_{[j:\omega+d]}$ belongs to $\mathcal{B}(n_1, n_1-d, d)$ with $n_1 = \omega+d - j$.
\end{proof}

With the introduced definitions, we present the main theorems of this paper below.
\begin{thm}\label{relation}
For any $\sni$ in $\mathcal{P}(n,\omega)$ and any integer $c$ with $\nsmall \leq c \leq \min\{\omega, \n+1\}$, there exists an integer $k$ such that $L^{k}(\textbf{s}_n) \in \mathcal{B}(n,c)$.
Furthermore,
$$\bigcup\limits_{\omega=c}^{n-1}\mathcal{P}(n,\omega)= \left\{(L^{k}(\textbf{s}_n))^{\infty} \,: \,\textbf{s}_{n}\in \mathcal{B}(n,c), \,  0\leq k<n \right\}.$$
\end{thm}

On the other hand, given a sequence $\textbf{s}_{n}$ in $\mathcal{B}(n,c,d)$ with $c \geq \lfloor\frac{n}{2}\rfloor$ and $add(\sn)=t$,
its companion pairs are of the form $\s_{c+d}$ as in \eqref{c+d}.
This property leads to the statement in Proposition \ref{fini-shift-nlc} (i), which implies
\begin{equation*}
add(R^{k}(\textbf{s}_{n}))=add(\textbf{s}_{n})-k \
\text{ and }\   nlc(R^{k}(\textbf{s}_{n}))=nlc(\sn)+k
\end{equation*}  for any $1\leq k\leq \min\{t, n-c-d\}$.
This observation
results in a more explicit expression of the nonlinear complexity of $\sni$ from  its $n$-length subsequences in $\mathcal{R}(n,c)$.

\begin{thm}\label{c+t}
For $\textbf{s}_{n}$ in $\mathcal{R}(n,c)$ with $c\geq \n$,  one has $nlc(\textbf{s}_n^{\infty})=nlc(\sn)+add(\sn)$.
\end{thm}

\subsection{Proof of Theorem \ref{relation} with $ c\leq \n+1$}

\noindent\textbf{Proof of Theorem \ref{relation}.}
For the periodic sequence $\textbf{a}_n^{\infty}$ in $\mathcal{P}(n, \omega)$,
suppose its left shift sequence $\textbf{s}_n^{\infty}=L^i(\textbf{a}_n^{\infty})$ has a companion pair $(\underline{\textbf{s}}_{0}, \calF^{d}({\underline{\textbf{s}}}_{0}))$ with $d\leq \n$,
we shall show that $L^{\omega-c}(\sn)$ belongs to  $\mathcal{B}(n,c)$.
According to Lemma \ref{lem_s_i^i+k+w} (ii),
if $\omega+d\leq n$ then the subsequence $\sn$ belongs to $\mathcal{B}(n,\omega,d)$.
This together with Lemma \ref{c-c+1} implies that $L^{\omega-c}(\textbf{s}_n)\in \mathcal{B}(n,c,d)$;
if $\omega+d> n$, then letting $j=\omega+d-n $ we have $L^{j}(\textbf{s}_n)\in \mathcal{B}(n,n-d,d)$, where $n-d\geq \left\lceil\frac{n}{2}\right\rceil$.
We need to consider two cases: $n-d\geq c$ and $n-d< c$.
For the case that $n-d\geq c$,
by applying Lemma \ref{c-c+1}, one has $L^{j+(n-d-c)}(\s_n)=L^{\omega-c}(\textbf{s}_n)\in \mathcal{B}(n,c,d)$;
for the case that $n-d< c$, since $n-d\geq \lceil\frac{n}{2}\rceil$, we have $c=\frac{n}{2}+1$ and $d=\frac{n}{2}$,
indicating $L^{j}(\textbf{s}_n)\in \mathcal{B}(n,\frac{n}{2},\frac{n}{2})$.
That is to say,
$$L^{j}(\textbf{s}_n)=(s_{j},\dots,s_{j+n-1}) =\left((s_{j},\dots,s_{j+\frac{n}{2}-2},s_{j+\frac{n}{2}-1})(s_{j},\dots,s_{j+\frac{n}{2}-2},\overline{s}_{j+\frac{n}{2}-1})\right).$$
Due to $nlc(\textbf{s}_n^{\infty})= \omega\geq c= \frac{n}{2}+1$, the subsequence $\s_{[j-1:\omega+d]}$ belongs to $\mathcal{B}(n+1, n+1-d, d)$ with $d=\frac{n}{2}$.
Thus it follows from \eqref{Eq-Relation} that $s_{j-1} = s_{j-1+d}$ and $s_{j-1+(n-d)} \neq s_{j-1+(n-d)+d} = s_{j-1 + n} = s_{j-1}$.
This is a contradiction since $s_{j-1+d} = s_{j-1 + (n-d)}$ for $d=\frac{n}{2}$.
Hence one always has $L^{\omega-c}(\textbf{s}_n)\in \mathcal{B}(n,c)$, the first statement follows.

Therefore,
for any sequence $\textbf{a}_n^{\infty}$ in $\mathcal{P}(n,\omega)$ and any integer $c$ with $\nsmall \leq c\leq \min\{\omega,\n+1\}$, we have
$$\textbf{a}_n^{\infty}\in \{(L^k(\textbf{s}_n))^{\infty}: \textbf{s}_n\in \mathcal{B}(n,c), 0\leq k<n \}.$$
In addition, it is clear that
$\left\{(L^{k}(\textbf{s}_n))^{\infty} \,: \,\textbf{s}_{n}\in \mathcal{B}(n,c), \,  0\leq k<n \right\} \subset$
$\bigcup\limits_{\omega=c}^{n-1}\mathcal{P}(n, \omega)$.
Thus
for $\nsmall \leq c \leq \min\{\omega, \n+1\}$, one has
$$\quad\quad\quad\quad\quad\quad\quad\quad
\bigcup\limits_{\omega=c}^{n-1}\mathcal{P}(n,\omega)= \left\{(L^{k}(\textbf{s}_n))^{\infty} \,: \,\textbf{s}_{n}\in \mathcal{B}(n,c), \,  0\leq k<n \right\}. \quad  \quad\quad\quad\quad \quad \text{\hfill} \square$$

Theorem \ref{relation} presents a one-to-one correspondence between the finite-length sequences set and
a set of all periodic sequences with nonlinear complexity not less than $c$.
For each periodic sequence $\sni$ with nonlinear complexity $\omega \geq c$,
Theorem \ref{relation} indicates the structure of an its $n$-length subsequence $\textbf{a}_n$, which is in the form of \eqref{structure of finite}.
According  to the structure of finite-length sequences $\textbf{a}_n$, we can obtain all periodic sequences with nonlinear complexities not more than $\frac{n}{2}$ in Subsection \ref{Sec4.1}.

When $c=\n$ or $\n+1$, from Corollary \ref{B_nlc} the nonlinear complexity of the $n$-length subsequence $\textbf{a}_n$ can be determined, i.e. $nlc(\textbf{a}_n)=c$.
In what follows,
we shall determine the exact value of nonlinear complexity $\omega$ of periodic sequences in $\{(L^{k}(\textbf{s}_n))^{\infty} \,: \,\textbf{s}_{n}\in \mathcal{B}(n,c),
\,  0\leq k<n \}$, where $\omega$ belongs to the set $[c,n-1]$ and $c\geq \n$.
With the further analysis, we can obtain all periodic sequences with nonlinear complexities not less than $\frac{n}{2}$ in Subsection \ref{Sec4.2}.

\subsection{Proof of Theorem \ref{c+t} with $ c\geq \n$}
Recall from \eqref{Eq-Rnc} that
$$\mathcal{R}(n,c) = \bigcup_{\sn\in \mathcal{B}(n,c)} \left\{\widetilde{\s}_n\in E(\sn) :\,  add(\widetilde{\s}_n) \geq add(\textbf{a}_n),   \forall\, \textbf{a}_n \in E(\sn) \right\},$$
where $E(\sn) = \{L^k(\sn):\, 0\leq k <n \} \cap \mathcal{B}(n,c)$.
Note that $\mathcal{R}(n,c)$ contains all cyclic shift inequivalent sequences in $\mathcal{B}(n,c)$.
From Lemma \ref{per-nlc} (ii), it suffices to investigate the nonlinear complexity of the periodic sequence $\sni$ derived from $\sn$ in $\mathcal{R}(n,c)$.

We are now ready to prove Theorem \ref{c+t}, namely,
for $\textbf{s}_{n}$ in $\mathcal{R}(n,c)$ with $c\geq \n$,  one has $nlc(\textbf{s}_n^{\infty})=nlc(\sn)+add(\sn)$.

\noindent\textbf{Proof of Theorem \ref{c+t}.}
Suppose $nlc(\textbf{s}_n^{\infty})=\omega$ and $(\underline{\textbf{s}}_{i}, \,\calF^{d}({\underline{\textbf{s}}}_{i}))$ is a companion pair of $\textbf{s}_n^{\infty}$
with $1\leq d\leq \lfloor\frac{n}{2}\rfloor$.
For the convenience of readers,
without loss of generality we will simplify $i$ into $0$ in the following proof.
It allows us to consider the companion pair $(\underline{\textbf{s}}_{0}, \,\calF^{d}({\underline{\textbf{s}}}_{0}))$ of $\textbf{s}_n^{\infty}$.
Below we will show that $\omega$ can be represented as $\omega = nlc(\textbf{a}_n) + add(\textbf{a}_n)$ for a concrete $n$-length subsequence $\textbf{a}_n$ of $\sni$ derived from $(\underline{\textbf{s}}_{0}, \,\calF^{d}({\underline{\textbf{s}}}_{0}))$, and then prove that  $\omega = nlc(\textbf{s}_n) + add(\textbf{s}_n)$ for $\sn \in \mathcal{R}(n, c)$. We divide the discussion according to the value of $\omega+d$.

{\bf Case (1)}:
$\omega+d< n$. In this case, Lemma \ref{lem_s_i^i+k+w} (ii) implies $\sn \in \mathcal{B}(n, \omega, d)$.
Note that $add(\sn)=0$, otherwise it follows from Proposition \ref{fini-shift-nlc} (i) and $n-\omega-d \geq 1$ that the nonlinear complexity of $R(\sn)$ is $\omega+1$, which is larger than  $nlc(\textbf{s}_n^{\infty})=\omega$, a contradiction.
By Definition \ref{def-t}
and Lemma \ref{c-c+1},
we see that $\textbf{a}_n=L^{\omega-c}(\textbf{s}_n)\in \mathcal{B}(n,c,d)$ and $add(\textbf{a}_n)=\omega-c$, which implies
$$
\omega=c+add(\textbf{a}_n)=nlc(\textbf{a}_n)+add(\textbf{a}_n).
$$
From the fact that $\sn\in \mathcal{R}(n,c)$, we have $add(\sn)\geq add(\textbf{a}_n)=\omega-c$. On the other hand, taking $add(\sn)=t$, from Remark \ref{mark1.1}, there is a pair of identical $(c+t-1)$-tuple subsequences with different successors in $\sni$. This implies
$c+t\leq \omega$.
Thus in this case we have $add(\sn)=add(\textbf{a}_n)=\omega-c$.

\begin{figure}[H]
\centering
\centering
\scalebox{0.9}{		
	\begin{tikzpicture}
		\node  at (-0.9,0)  {  $ \s_{\omega+d} = (s_0,s_{1}, \dots,s_{j-1})  (s_{j},s_{j+1},\dots,s_{n},\dots,
			s_{\omega+d-1})  $ };
		\fill[gray, opacity=0.5] (-4.15,-0.3) rectangle (-1.3,0.3) ;
		\fill[gray, opacity=0.5] (1,-0.3) rectangle (3.75,0.3) ;
		\draw[<->,thick](-1.2,-0.38)--(3.75,-0.38) node[below] at (1.4,-0.4) {length=$n$};	
	\end{tikzpicture}
}
\caption{ The description of Case (2) in Theorem \ref{c+t}}\label{fig 6}
\end{figure}

{\bf Case (2)}: $\omega+d\geq n$. In this case, we know the subsequence $\s_{\omega+d} $ satisfies
$$
s_{i} = s_{i+d} \, \, \text{ for } \, \, 0\leq i \leq \omega-2  \, \, \text{ and }\, \, s_{i} = s_{i+n}\, \, \text{ for }\, \, i\geq 0.
$$
Take $j=\omega+d-n$ 
and 
let $\textbf{a}_n = L^j(\sn)$. 
Then $\textbf{a}_n=(a_0,a_1,\dots, a_{n-1})=(s_j, s_{j+1}, \dots, s_{j+n-1})$
$=(s_j, s_{j+1}, \dots, s_{\omega+d -1})$ in Figure \ref{fig 6}. It follows from Lemma \ref{lem_s_i^i+k+w} (iii) that $\textbf{a}_n \in \mathcal{B}(n,n-d,d)$.
Since $n-d \geq \lceil\frac{n}{2}\rceil$, $\textbf{a}_n$ has nonlinear complexity $n-d$ by Corollary \ref{B_nlc}. From the above equalities,
the sequence $\textbf{a}_n$ satisfies that for $0\leq j_1< j$,
$$
a_{(n-1)-j_1} = s_{(j+n-1)-j_1}= s_{(j-1)-j_1} = s_{(j-1)-j_1+d}=a_{(d-1)-j_1},
$$
where $j=\omega+d-n<d$ and $j-j_1\leq d-1\leq \omega-1$.
Since $s_{n-1}=s_{-1} \neq s_{d-1}$,
according to Definition \ref{def-t}, we have
$$add(L^j(\sn)) = add(\textbf{a}_n) = j=\omega+d-n = \omega - (n-d).$$
This implies $\omega=nlc(\textbf{a}_n)+add(\textbf{a}_n)$.

Consider the following set
$$S=\left\{L^{k}(\textbf{s}_{n})\, : \, L^{k}(\textbf{s}_{n})\in \mathcal{B}(n,c'), \ c'\geq  \frac{n}{2} \ \mbox{and} \ 0\leq k < n\right\}.$$
We see that $\omega$ can be represented as the form $nlc(\textbf{a}_n)+add(\textbf{a}_n)$ for $\textbf{a}_n=L^j(\sn)$, $j=\omega+d-n$, in $S$.
By Remark \ref{mark1.1} we know that for any $\textbf{u}_n \in S$, $\textbf{u}_n^{\infty}$ contains a pair of identical subsequences with length $nlc(\textbf{u}_n)+add(\textbf{u}_n)-1$ with different successors.
Thus we have $\omega=\max\limits_{\textbf{u}_n\in S} (nlc(\textbf{u}_n)+add(\textbf{u}_n))$.
In the following we shall show $\omega=c+t$ for three subcases.

{\bf Subcase (2.1)}:
Consider the set
$$S_{0}=\{L^{k}(\textbf{s}_{n})\, : \,  \,  0\leq k < n\} \cap \mathcal{B}(n,c)=E(\textbf{s}_{n}).$$
It follows that $\max\limits_{\textbf{u}_n\in S_{0}} (nlc(\textbf{u}_n)+add(\textbf{u}_n))=c+\max\limits_{\textbf{u}_n\in S_{0}}(add(\textbf{u}_n))=c+add(\textbf{s}_{n})= c+t$
since $\textbf{s}_{n}$ is a sequence representative of $E(\textbf{s}_{n})$ with maximal $add(\textbf{s}_{n})=t$ in the set $S_{0}$.

{\bf Subcase (2.2)}: Consider $S_{1}=\{L^{k}(\textbf{s}_{n})\, : \, L^{k}(\textbf{s}_{n})\in \mathcal{B}(n,c_{1}), \,  0\leq k < n\}$ with $c_1=c+t_1$ and $t_1 \geq 1$.
It follows that $c_{1} +\max\limits_{\textbf{u}_n\in S_{1}}(add(\textbf{u}_{n}))=c+t_{1} +\max\limits_{\textbf{u}_{n}\in S_{1}}(add(\textbf{u}_{n})) $.
As shown in Lemma \ref{c-c+1}, for each $\textbf{u}_{n} \in S_{1}$, one can get $L^{t_{1}}(\textbf{u}_{n})\in S_{0}$. In addition, since
\begin{align*}
S_{1}
&=\{L^{k}(\textbf{s}_{n})\, : \, L^{k}(\textbf{s}_{n})\in\mathcal{B}(n,c+t_{1}), \,  0\leq k < n\} \\
&=\{L^{k}(\textbf{s}_{n})\, : \, L^{k+t_{1}}(\textbf{s}_{n})\in\mathcal{B}(n,c), \, add(L^{k+t_{1}}(\textbf{s}_{n}))\geq t_{1}, \, 0\leq k < n\} \\
&=\{R^{t_{1}}(L^{k}(\textbf{s}_{n}))\, : \, L^{k}(\textbf{s}_{n})\in\mathcal{B}(n,c), \, add(L^{k}(\textbf{s}_{n}))\geq t_{1}, \, 0\leq k < n\},
\end{align*}
it follows that $\max\limits_{\textbf{u}_{n}\in S_{1}}(add(\textbf{u}_{n}))=\max\limits_{\textbf{u}_{n}\in S_{0}}(add(\textbf{u}_{n}))-t_{1}=t-t_{1}$.
This yields $\max\limits_{\textbf{u}_{n}\in S_{1}} (nlc(\textbf{u}_{n})+add(\textbf{u}_{n}))=c+t_{1}+(t-t_{1})=c+t$.

{\bf Subcase (2.3)}: Consider $S_{2}=\{L^{k}(\textbf{s}_{n})\, : \, L^{k}(\textbf{s}_{n})\in\mathcal{B}(n,c_{2}), \,  0\leq k < n\}$ with $\frac{n}{2} \leq c_2 < c$.
Take $c_{2}=c-t_{2}$ with $t_2 \geq 1$ . We shall investigate
$c-t_{2}+\max\limits_{\textbf{u}_{n}\in S_{2}}(add(\textbf{u}_{n}))$ for $S_2$.
For each $ \textbf{u}_{n}\in S_{0}$, it follows from Lemma \ref{c-c+1} that $L^{t_{2}}(\textbf{u}_{n})\in S_{2}$. In a similar manner, one has
\begin{align*}
S_{0}
&=\{L^{k}(\textbf{s}_{n})\, : \,  L^{k}(\textbf{s}_{n})\in\mathcal{B}(n,c), \,  0\leq k < n\} \\
&=\{L^{k}(\textbf{s}_{n})\, : \,  L^{k+t_{2}}(\textbf{s}_{n})\in\mathcal{B}(n,c_{2}),
\,  add(L^{k+t_{2}}(\textbf{s}_{n}))\geq t_{2}, \, 0\leq k < n\} \\
&=\{R^{t_{2}}(L^{k}(\textbf{s}_{n}))\, : \,L^{k}(\textbf{s}_{n})\in\mathcal{B}(n,c_{2}), \,  add(L^{k}(\textbf{s}_{n}))\geq t_{2}, \, 0\leq k< n\}.
\end{align*}
Then $\max\limits_{\textbf{u}_{n}\in S_{2}}(add(\textbf{u}_{n}))-t_{2}=\max\limits_{\textbf{u}_{n}\in S_{0}}(add(\textbf{u}_{n}))$.
This yields $ \max\limits_{\textbf{u}_{n}\in S_{2}} (nlc(\textbf{u}_{n})+add(\textbf{u}_{n}))=c-t_{2}+\max\limits_{\textbf{u}_{n}\in S_{2}}(add(\textbf{u}_{n}))=c-t_{2}+(t+t_{2})=c+t$.

Combining the above three subcases, we have $\omega=\max\limits_{\textbf{u}_{n}\in S} (nlc(\textbf{u}_{n})+add(\textbf{u}_{n}))=c+t$. The desired conclusion thus follows.
\hfill $\square$

\begin{remark}\label{Repredsentation}
With the condition in Theorem \ref{c+t}, it follows that each sequence $\textbf{a}_{n}$ in the set $E(\textbf{s}_{n})$ have $nlc(\textbf{a}_{n}^{\infty})=nlc(\sni)=nlc(\sn)+add(\sn)$ by Lemma \ref{per-nlc} (ii).
Equivalently, 		
for any sequence $\text{a}_n$ with nonlinear complexity  $c\geq \frac{n}{2}$, 
by Lemma \ref{lem unique}
it contains a unique companion pair $(\textbf{\underline{a}}_i, \widehat{{\underline{\textbf{a}}}_{i}})$.
This implies $L^{i}(\textbf{a}_n)\in\mathcal{B}(n,c)$. Hence we have $nlc(\textbf{a}_{n}^{\infty})=nlc(\sni)=nlc(\sn)+add(\sn)$, where $\sn$ is a representative of $E(\textbf{a}_n)$.
\end{remark}

Theorem \ref{c+t} gives a method to determine the value of the nonlinear complexity of periodic sequences from that of finite-length sequences. Based on Theorem \ref{c+t}, below we give a corollary of the value of $
\max\limits_{0\leq i <n} nlc(R^i(\sn))$.

\begin{cor}\label{cor condition}
For $\textbf{s}_{n}$ in $\mathcal{R}(n,c,d)$ with $c\geq \n$ and $add(\sn)=t$,  one has
$$
\max\limits_{0\leq i <n} nlc(R^i(\sn)) \begin{cases}
	=c+t, & \text{ if } t\leq n-c-d,\\
	\geq n-d, & \text{ if } t> n-c-d.\\
\end{cases}
$$
\end{cor}
\begin{proof}
It follows from Theorem \ref{c+t} that for $\textbf{s}_{n}\in \mathcal{R}(n,c,d)$ with $add(\sn)=t$, we have  $nlc(\sni)=c+t$.
Hence $
\max\limits_{0\leq i <n} nlc(R^i(\sn)) \leq nlc(\sni)=c+t$.
According to Proposition \ref{fini-shift-nlc} (i),
if $ t\leq n-c-d$, i.e. $c+d+t\leq n$,
then $R^{t}(\sn) \in \mathcal{B}(n,c+t,d)$ with nonlinear complexity $c+t$, implying $\max\limits_{0\leq i <n} nlc(R^i(\sn)) =c+t=nlc(\sni)$;
if $ t> n-c-d$, i.e. $c+d+t> n$, then one has $R^{n-c-d}(\sn) \in \mathcal{B}(n,n-d,d)$ with nonlinear complexity $n-d$, thus $
\max\limits_{0\leq i <n} nlc(R^i(\sn)) 	\geq n-d$.
\end{proof}

Based on Corollary \ref{cor condition}, we can see that
for $\textbf{s}_{n}$ in $\mathcal{R}(n,c,d)$ with $c\geq \n$ and $add(\sn)=t$,
if $ t\leq n-c-d$, then $\max\limits_{0\leq i <n} nlc(R^i(\sn))=nlc(\sni)$.
That is to say, Corollary \ref{cor condition} gives a sufficient condition of
$\max\limits_{0\leq i <n} nlc(R^i(\sn))=nlc(\sni)$.
When $t> n-c-d$, numerical results for $n\leq 30$ indicate that $\max\limits_{0\leq i <n} nlc(R^i(\sn))=n-d$. However, the previous technique is not sufficient to
prove or disprove this equality. 
If the conjecture is true, then we can obtain a sufficient and necessary condition of
$\max\limits_{0\leq i <n} nlc(R^i(\sn))=nlc(\sni)$.
Here we propose an open problem below on this observation and cordially invite interested readers to attack this problem.

\begin{Open problem}
For $\textbf{s}_{n}$ in $\mathcal{R}(n,c,d)$ with $c\geq \n$ and $add(\sn)=t$,
is it true that
$	\max\limits_{0\leq i <n} nlc(R^i(\sn))=n-d$ when $t> n-c-d$?
\end{Open problem}

\section{Periodic sequences with prescribed nonlinear complexity}\label{Sec4}
Based on new theoretical results in Theorem \ref{relation} and Theorem \ref{c+t},
we shall give two algorithms to generate all periodic binary sequences in $\mathcal{P}(n,\omega)$ with any prescribed nonlinear complexity.
For every shift equivalence class $\{\textbf{s}_n^{\infty}, (L(\textbf{s}_n))^{\infty}, \dots, (L^{n-1}(\textbf{s}_n))^{\infty}\}$,
since each sequence in the class can generate the whole class,  it suffices to consider one sequence with certain property when generating the shift equivalence class.
Given a set $S$ of sequences with period $n$, we denote by $\overline{S}$ the union of all cyclic shift equivalence classes of sequences in $S$.
Then two sets $S_1$ and $S_2$ are deemed to be cyclic shift equivalent, denoted as $S_1 \cong S_2$,
if $\overline{S_1}=\overline{S_2}$.

\begin{thm}\label{thm_core}
Let $n$ and $\omega$ be two positive integers with $\lceil \log_{2}(n)\rceil \leq \omega\leq n-1$.
Let
$\mathcal{B}(n,c)$ be given in \eqref{Eq-Bnc}, $\mathcal{B}_0(n,c)=\{\sn \in \mathcal{B}(n,c):  add(\sn)=0\}$
and $\mathcal{R}(n,c)$ be defined by \eqref{Eq-Rnc}, respectively.
Then we have
\begin{equation*}
	\mathcal{P}(n,\omega)
	\cong \{\sni: \sn\in S(n,\omega)\}
\end{equation*}
where $S(n,\omega)$ is given as follows,

\noindent{\rm (i)} if $\omega\leq \frac{n}{2}$, then
$S(n,\omega)= \{\sn \in\mathcal{B}_0(n,\omega):    L^{k}(\sn) \notin \mathcal{B}(n,\omega+1),  0< k<n   \};$

\noindent{\rm (ii)}
if $\omega\geq \frac{n}{2}$, then $S(n,\omega)= \{  \sn \in \mathcal{R}(n,c) : \, add(\sn)=\omega-c  \}$, where $c=\lceil\frac{n}{2}\rceil$.
\end{thm}
\begin{proof}
(i) According to Theorem \ref{relation}, for any integer $c$ with $\nsmall \leq c \leq \min\{\omega, \n+1\}$, we have
$$\bigcup\limits_{\omega=c}^{n-1}\mathcal{P}(n,\omega)= \left\{(L^{k}(\textbf{s}_n))^{\infty} \,: \,\textbf{s}_{n}\in \mathcal{B}(n,c), \,  0\leq k<n \right\}.$$
Then for $\lceil \log_{2}n \rceil \leq \omega\leq \n$, we can express
\begin{align*}
	&\mathcal{P}(n,\omega)\\
	=	&\bigcup\limits_{\omega_{1}=\omega}^{n-1}\mathcal{P}(n,\omega_{1}) \Big\backslash \bigcup\limits_{\omega_{2}=\omega+1}^{n-1}\mathcal{P}(n,\omega_{2})\\
	=&\left\{(L^{k}(\textbf{s}_n))^{\infty} \,:\,\textbf{s}_{n}\in \mathcal{B}(n,\omega), \,  0\leq k<n \right\} \Big\backslash\left\{(L^{k}(\textbf{s}_n))^{\infty} \,: \,\textbf{s}_{n}\in \mathcal{B}(n,\omega+1), \,  0\leq k<n \right\}\\
	\cong&\left\{\textbf{s}_n^{\infty} \,: \sn \in\mathcal{B}(n,\omega), L^{k}(\sn) \notin \mathcal{B}(n,\omega+1),  0< k<n   \right\}.
\end{align*}
For $\sn$ in the set $\left\{ \sn \in\mathcal{B}(n,\omega) : L^{k}(\sn) \notin \mathcal{B}(n,\omega+1),  0< k<n   \right\}$, we have $s_{n-1} \neq s_{d-1} $, that is to say $add(\sn)=0$, which implies  $\sn \in\mathcal{B}_0(n,\omega) $.
Thus for $\lceil \log_{2}(n) \rceil\leq \omega\leq \lfloor\frac{n}{2}\rfloor$,
$$
\mathcal{P}(n, \omega)\cong \left\{\textbf{s}_n^{\infty} \,: \textbf{s}_{n}\in \mathcal{S}(n, \omega) \right\},
$$ where
$\mathcal{S}(n,\omega) = \{ \sn \in\mathcal{B}_0(n,\omega) : L^{k}(\sn) \notin \mathcal{B}(n,\omega+1),  0< k<n   \}$.
As a result, the desired statement of Theorem \ref{thm_core} (i) follows.

(ii) Again by Theorem \ref{relation}, let $c= \nup$ and $\omega\geq c$, we have
$$\bigcup\limits_{\omega=c}^{n-1}\mathcal{P}(n,\omega)= \left\{(L^{k}(\textbf{s}_n))^{\infty} \,: \,\textbf{s}_{n}\in \mathcal{B}(n,c), \,  0\leq k<n \right\}\cong \{ \sni: \sn \in \mathcal{R}(n,c)\}$$
since $\mathcal{R}(n,c)$ contains all cyclic shift inequivalent sequences in $\mathcal{B}(n,c)$.
Then by Theorem \ref{c+t} we have
$\mathcal{P}(n,\omega)\cong \{ \sni: \sn \in \mathcal{R}(n,c) : \, add(\sn)=\omega-c\}
= \{\textbf{s}_n^{\infty} \ : \  \textbf{s}_{n}\in S(n,\omega)\}.$
\end{proof}

Based on Theorem \ref{thm_core} (i) and (ii),  the following two subsections are dedicated to the generation of periodic binary sequences with any given nonlinear complexity.

\subsection{Periodic sequences with nonlinear complexity $\leq \frac{n}{2}$}\label{Sec4.1}

\begin{algorithm}
\small
\caption{Generation of all periodic binary sequences in $ \mathcal{P}(n,\omega)$ with $\omega\leq \frac{n}{2}$ }\label{alg:period-small}
\begin{algorithmic}[1]	
	\State    INPUT: Given two global variables $n$ and $\omega$ with $ \lceil \log_{2}(n)\rceil  \leq \omega\leq \n$.
	\State    OUTPUT: The set $\mathcal{P}(n,\omega)$.
	
	\State\textbf{Main Algorithm}
	\State $\mathcal{B}_0(n,\omega)\gets$  \texttt{genB}($n,\omega$)
	\State $\mathcal{B}(n,\omega+1)\gets$ \texttt{genB}($n,\omega+1$)
	\State $\mathcal{S}(n,\omega)= \{ \sn \in\mathcal{B}_0(n,\omega) : L^{k}(\sn) \notin \mathcal{B}(n,\omega+1),  0< k<n   \}$  
	\State  $\mathcal{P}(n, \omega)=\left\{(L^k(\textbf{s}_n))^{\infty} \,: \textbf{s}_{n}\in \mathcal{S}(n, \omega), 0\leq k<n \right\}$	
	\Function {$\mathtt{genB}$}{$n,c$} \hfill // Generate $\mathcal{B}_0(n,\omega)$ and $\mathcal{B}(n,\omega+1)$  
	\For {$d=1$ to $\n$}
	\State $\mathcal{B}(n,c,d)\gets\emptyset$   
	\While {$(a_{0},a_1,\dots,a_{d-1})\in\mathbb{Z}_2^{d}$ is aperiodic}
	\For {$i=0$ to $c+d-2$}
	\State $s_{i}\gets a_{i\, \text{mod}\,\,d}$
	\EndFor
	\State $s_{c+d-1}\gets a_{(c+d-1)\,\text{mod}\,\,d\,}\oplus1$
	\For{$(s_{c+d},\dots, s_{n-2})\in\mathbb{Z}_2^{n-c-d-1}$}
	\State $\textbf{s}_{n-1}\gets (s_0, \dots, s_{c+d-1}, s_{c+d}, \dots, s_{n-2})$
	\If{$c=\omega$}
	\State $\textbf{s}_{n}\gets (\s_{n-1}, \overline{s}_{d-1})$
	\Else
	\State $\textbf{s}_{n}\gets (\s_{n-1}, s_{n-1})$ with $s_{n-1} \in \mathbb{Z}_2$
	\EndIf
	\If{$c\geq \lfloor\frac{n}{2}\rfloor$ or $n$ is prime}  \hfill// Conditions that $\sn$ is aperiodic
	\State Add the sequence $\textbf{s}_{n}$ to $\mathcal{B}(n,c,d)$
	\ElsIf{$\textbf{s}_{n}$ is aperiodic}
	\State Add the sequence $\textbf{s}_{n}$ to $\mathcal{B}(n,c,d)$
	\EndIf
	\EndFor
	\EndWhile
	\EndFor	
	\State \Return $\bigcup_{1 \leq d \leq \n} \mathcal{B}(n,c,d)$ 
	\EndFunction
	
\end{algorithmic}
\end{algorithm}
In this subsection, all periodic sequences in $\mathcal{P}(n,\omega)$ with $\lceil \log_{2}(n) \rceil \leq \omega\leq \lfloor\frac{n}{2}\rfloor$ can be generated.
By summarizing up the preceding analysis, we propose an algorithm to generate all periodic binary sequences in $\mathcal{P}(n,\omega)$ for any $\lceil \log_{2}(n)\rceil\leq \omega \leq \n$.

In Algorithm \ref{alg:period-small}, we provide detailed steps to generate the set $\mathcal{B}(n, c)$ for any $1\leq c<n$.
The function $\mathtt{genB}(n,\omega+1)$ in Algorithm \ref{alg:period-small} for $\omega+1 \geq \frac{n}{2}$ is similar to the algorithm in \cite{liang}.
For $\lceil \log_{2}(n)\rceil \leq\omega \leq \n$, Algorithm \ref{alg:period-small} generates all periodic sequences in $\mathcal{P}(n,\omega)$ based on Theorem \ref{thm_core} (i).

We now consider the complexity of Algorithm \ref{alg:period-small}, where Steps 4, 5, 6 dominate the complexity.
In generating $\mathcal{B}_0(n, \omega)$, for each $1\leq d\leq \lfloor\frac{n}{2}\rfloor$ the loops on $(a_0, \dots, a_{d-1})$ and  $(s_{\omega+d}, \dots, s_{n-2})$  contribute the time and memory complexity $O(2^{n-\omega-1})$,
implying Step 4 has complexity $O(n2^{n-\omega-2})$ as $d$ ranges from $1$ to $\lfloor\frac{n}{2}\rfloor$.  Similarly Step 5 has complexity $O(n2^{n-\omega-2})$.
Step 6 generates the set
\begin{align*}
\mathcal{S}(n,\omega)=& \{ \sn \in\mathcal{B}_0(n,\omega) : L^{k}(\sn) \notin \mathcal{B}(n,\omega+1),  0< k<n   \}\\
=& \left\{ L^{k}(\textbf{s}_n)\,:\,\textbf{s}_{n}\in \mathcal{B}_0(n,\omega), \,  0\leq k<n \right\} \Big\backslash\mathcal{B}(n,\omega+1).
\end{align*}
By using the data structure of hash table,  this step will have time and memory complexity $O(\min\{ n|\mathcal{B}_0(n,\omega)|, |\mathcal{B}(n,\omega+1)| \})$, where
$|S|$ is denoted as the size of a set $S$.  Since  $|\mathcal{B}_0(n,\omega)| < 2^{n-\omega-1}$, the complexity of Algorithm \ref{alg:period-small} can be obtained as
$O(n2^{n-\omega-2}) +O(n2^{n-\omega-2}) + O( n|\mathcal{B}_0(n,\omega)|  ) \approx
O(n2^{n-\omega-1})$.

On the other hand, when one uses a brute-force approach to generating sequences of given nonlinear complexities,
for each $n$-periodic binary sequence, it is required to calculate its nonlinear complexity with algorithms like the ones proposed in \cite{JansenPhD, LKK2}, which have complexity $O(n^{2} \log_2(n))$.
On average, we may consider the complexity of exhaustive search for all periodic sequences in $\mathcal{P}(n,\omega)$ is $\frac{2^{n}O(n^{2} \log_2(n))}{n} \approx O(n\log_2(n)2^n)$.
This indicates that Algorithm \ref{alg:period-small} has an advantage of factor $2^{\omega+1}\log_2(n)$ over an exhaustive search for periodic binary sequences in $\mathcal{P}(n,\omega)$.
It is apparent that such an advantage becomes significant as $\omega$ increases.

Below we provide an example to illustrate the process and result of Algorithm \ref{alg:period-small}.

\begin{example}
Take $n=7$ and $\omega = 3$. By the function $\mathtt{genB}(n,c)$ in Algorithm \ref{alg:period-small}, we generate the set  $\mathcal{B}_0(n, \omega)$ consisting of $18$ binary sequences of length $7$.
Among these sequences, these twelve sequences
$( 0001101)$, $( 0110100)$,$(1010001)$,
$(0001011)$,$(0101100)$, $( 1011000 )$,
$(1110100)$,  $(1010011)$,  $(0100111)$,
$(1110010)$,$(1001011)$,$(0101110)$,
satisfy the property that all their cyclic shift equivalent sequences (except for themselves) are not contained in $\mathcal{B}(n, \omega+1)$. That
is to say, the above twelve sequences form the set $\mathcal{S}(n, \omega)$. On the other hand, by exhaustive search we obtain the following sequences in $\mathcal{P}(n,\omega)$:
$$
\begin{array}{c}
	( 0001101)
	(0011010 )
	( 0110100 )
	(110100 0 )
	( 10100 0 1 )
	( 0100 0 1 1 )
	( 100 0 1 1 0  )\\
	
	( 0001011)
	(0010110 )
	( 010110 0 )
	(10110 0  0 )
	( 0110 0  01 )
	( 110 0  010 )
	( 10 0  0101  )\\	
	
	( 1110100)
	(1101001 )
	( 101001 1 )
	(01001 11 )
	( 1001 110 )
	( 001 110 1 )
	( 01 110 1 0  )\\		
	
	( 1110010 )
	( 110010  1  )
	( 10010  11  )
	(0010  11 1)
	(010  11 10)
	(10  11 100)
	(0  11 1001)\\
	
\end{array}
$$
It is easily seen that $\mathcal{P}(n,\omega)$ can be obtained by applying circular shift operations on sequences in $\mathcal{S}(n, \omega)$, i.e., $\mathcal{P}(n, \omega) = \{(L^k(\sn))^\infty: \, \sn \in \mathcal{S}(n, \omega), \, 0\leq k<n\}$.	
\end{example}

\subsubsection{Generation of binary de Bruijn sequences}
Here we consider a particular case of $n=2^{m}$ and $\omega=m$,
which corresponds to the generation of binary de Bruijn sequences of order $m$.
In the literature, graphical, algorithmic and algebraic approaches have been proposed to generate binary de Bruijn sequences \cite{Fredricksen_1982_, Helleseth-Li}.
The graphical approach, known as the BEST theorem for de Bruijn, Ehrenfest, Smith and
Tutte, showed that there are in total $2^{2^{m-1}-m}$ binary de Bruijn sequences of order $m$ and they can be derived in an inductive manner \cite{Fredricksen_1982_}.
Algorithmic and algebraic approaches can be used to generate some de Bruijn sequences. However, to the best of our knowledge,
only the graphical approach can constructively generate all de Bruijn sequences of order $m$.

For the case of $n=2^m$ and $\omega = m$, we can adjust some steps in Algorithm \ref{alg:period-small} for better performance.
For a binary periodic sequence,
a run of $0$'s of length $k$ is a string of consecutive $k$ $0$'s flanked by 1 and a run of 1's of length $k$ is a string of consecutive $k$ $1$'s flanked by 0.
It is well-known that any binary de Bruijn sequence of order $m$ satisfies run properties as below \cite{Golomb_1980_}:
\begin{equation}\label{run-prop}
\begin{cases}
	\text{(i) } 2^{m-2-i}  \text{ runs of zeros, }  2^{m-2-i}  \text{ runs of ones of length }  i, \text{ for } 1 \leq i < m -1;\\
	\text{(ii)}   \text{ no run of zeros nor ones of length } m-1;\\
	\text{(iii)} \text{ a single run of }   m \text{ zeros and a single run of }  m  \text{ ones}. 	
\end{cases}
\end{equation}
From the above three properties, it is easily seen that
any binary de Bruijn sequence of order $m$ has Hamming weight $2^{m-1}$ and contains the subsequence $(1,\overbrace{0,\dots, 0}^{m},1)$, which has the companion pair $({0}^m, {0}^{m-1}1)$
with spacing $d=1$.
With this observation, we can fix the sequence $\s_{c+d}$ in Algorithm \ref{alg:period-small}  as $\s_{c+d}=\s_{m+1}=(\overbrace{0,\dots 0,}^{m}1)$ and $s_{n-1}=1$, and focus on only the binary sequences $(s_{m+1},\dots, s_{n-2})$ of Hamming weight $2^{m-1}-2$. Thus we can reduce $\mathcal{B}_0(n,\omega)$ to a smaller set
\begin{equation}\label{Eq_Bnm}
	 \widetilde{\mathcal{B}}_0(n,m) = \{\sn=(\overbrace{0,\dots, 0}^{m},1, s_{m+1},\dots, s_{n-2}, 1)\,:\, \sn \text{ satisfies \eqref{run-prop}}\},
\end{equation}
which is generated in Algorithm \ref{alg:period-m}.

\begin{algorithm}
	\small
	\caption{Generation of binary sequences in $ \widetilde{\mathcal{B}}_0  (2^m,m)$ }\label{alg:period-m}
	\begin{algorithmic}[1]
		\Function {$\mathtt{gen \widetilde{{B}}_0 }$}{$n,m$} \hfill // Generate $ \widetilde{\mathcal{B}}_0  (2^m,m)$
		\State $\textbf{s}_{m+1}\gets(0,0,\dots,0,1)$ and $ \widetilde{\mathcal{B}}_0(n,m) \gets\emptyset$
		\For{$(s_{m+1},\dots, s_{n-2})\in\mathbb{Z}_2^{n-m-2}$ with Hamming weight $2^{m-1}-2$}
		\If{$\textbf{s}_{n}=({0}^{m}1, s_{m+1},\dots, s_{n-2}, 1)$ satisfies the run properties \eqref{run-prop}}
		\State Add the sequence $\textbf{s}_{n}$ to $ \widetilde{\mathcal{B}}_0(n,m) $
		\EndIf
		\EndFor	
		\EndFunction
	\end{algorithmic}
\end{algorithm}

\begin{cor}\label{cor3}
Let  $n=2^{m}$ with a positive integer $m$ and let $ \widetilde{\mathcal{B}}_0(n,m)  $ be as in \eqref{Eq_Bnm}. Then,
we obtain de Bruijn sequences of order $m$ as follows,
$$\mathcal{P}(n, m)=\left\{(L^k(\textbf{s}_n))^{\infty}: \textbf{s}_{n}\in \mathcal{S}(n, m), 0\leq k<n \right\}$$
where
$\mathcal{S}(n,m)= \{ \sn \in \widetilde{\mathcal{B}}_0(n,m)  : L^{k}(\sn) \notin \mathcal{B}(n,m+1),  0< k<n   \}$
and $| \widetilde{\mathcal{B}}_0(n,m)  |=\frac{1}{2^{2m-2}} ( \frac{2^{m-2}}{ 2^{m-3}, \, 2^{m-4}, \dots,   2^{1}, \,2^{0} })^2$ with $(\frac{ M}{ m_1,\, m_2, \dots, m_k})
=\frac{ M!}{ m_1!\, m_2! \, \dots m_k !}$.
\end{cor}
\begin{proof}
According to the above analysis, we can
reduce $\mathcal{B}_0(n,\omega)$ to $ \widetilde{\mathcal{B}}_0(n,m)  $ in \eqref{Eq_Bnm}.
Thus
$\mathcal{P}(n, m)$ can be obtained immediately from Theorem \ref{thm_core} (i).
In the following we determine the
size of $ \widetilde{\mathcal{B}}_0(n,m) $.
A binary sequence of period $2^m-1$ having the same run distribution as $m$-sequences of order $m$ is called as a run sequence.
Note that sequences in $ \widetilde{\mathcal{B}}_0(n,m)  $ are in one-to-one correspondence with all binary run sequences of period $2^m-1$.
For each sequence in $ \widetilde{\mathcal{B}}_0(n,m)  $,
by deleting one $0$ in the longest $m$ run of $0$'s, we obtain a run sequence of period $2^m-1$.
Hence, the number of $ \widetilde{\mathcal{B}}_0(n,m)  $ is the same as the number of shift inequivalent run sequences, from   \cite{Kim_2023_} which is equal to $ \frac{1}{2^{2m-2}} ( \frac{2^{m-2}}{ 2^{m-3}, \, 2^{m-4}, \dots,   2^{1}, \,2^{0} })^2$.
\end{proof}

\begin{example}
When $m=4$,
Algorithm \ref{alg:period-m}
generates $ \widetilde{\mathcal{B}}_0  (16, 4)$ with size $36$, then we filter these $36$ shift inequivalent sequences by $\mathcal{B}(16, 5)$, thus obtaining
all $16$ shift inequivalent sequences in
$\mathcal{P}(16, 4)$.
And when $m=5$,
Algorithm \ref{alg:period-m}
generates $ \widetilde{\mathcal{B}}_0  (32, 5)$ with size $88200$, then we filter these shift inequivalent sequences by $\mathcal{B}(32, 6)$, which yields
all $2048$ shift inequivalent sequences in
$\mathcal{P}(32, 5)$.
\end{example}

\subsection{Periodic sequences with  nonlinear complexity $\geq \frac{n}{2}$}\label{Sec4.2}
Recall from Theorem \ref{thm_core} (ii) that,  for $\omega\geq \frac{n}{2}$,
$$\mathcal{P}(n, \omega)\cong \left\{\textbf{s}_n^{\infty} \,: \textbf{s}_{n}\in \mathcal{S}(n, \omega) \right\},$$ where
$\mathcal{S}(n,\omega)$ for $c= \nup$ is given by $S(n,\omega)= \{  \sn \in \mathcal{R}(n,c) : \, add(\sn)=\omega-c  \}.$
This result also holds for $c=\n$ and can be proved similarly. Moreover,  the characterizations on sequences $\mathcal{R}(n, \lceil \frac{n}{2}\rceil)$ in Subsection 5.2.1 can
be similarly made for $c=\n$ but with several tedious cases to be discussed. We therefore only present relevant results for the case of $c=\lceil \frac{n}{2}\rceil$.

To illustrate the result of Theorem \ref{thm_core} (ii), we first present a toy example below.
\begin{example}
Take an example for the case $n=8$ and $c=4$.
All finite-length sequences in $\mathcal{B}(8,4)$ can be obtained from the function $\mathtt{genB}(n,c)$ in Algorithm \ref{alg:period-small},
or from the algorithm in \cite{liang} by letting $k_1=0$.  We first group $\mathcal{B}(8,4)$ into shift equivalence classes. For each shift equivalence class in $\mathcal{B}(8,4)$, we determine the sequence representatives according to Definition \ref{ES_Def}, thereby obtain the set $\mathcal{R}(8,4)$. In this way we get $\bigcup\limits_{\omega=4}^{7}\mathcal{P}(8,\omega)\cong \{ \sni: \sn \in \mathcal{R}(8,4)\}$.
By computing the number of added terms of each sequence in $\mathcal{R}(8,4)$, we can determine the nonlinear complexity of the corresponding periodic sequences from Theorem \ref{c+t}.
Table 2 lists all binary sequences with period $8$ and nonlinear complexity $\omega$, $\omega\geq 4$, up to shift equivalence obtained in this way.
The result is consistent with the exhaustive search presented in \cite[Table~3.2]{JansenPhD}.

We carry out experiments for $n$ up to 40, and they all confirm that  Theorem \ref{thm_core} (ii) is consistent with exhaustive search. 	
\begin{table}[!ht]\label{tabbb2}
	\scriptsize
	\caption{Binary sequences of period $8$ with nonlinear complexity $\omega$.}
	\begin{center}
		\begin{tabular}{|c|c|c|c|c| p{1cm}|}
			\hline
			$\mathcal{B}(8,4)$ & $\mathcal{R}(8,4)$ & $add(\textbf{s}_{n})$ & $\cong\mathcal{P}(8,4+add(\textbf{s}_{n}))$ \\ 
			\hline
			\hline
			(00100011), (10010001), (00110010)  & (00100011), (10010001), (00110010) & \multirow{10}{*}{0} &$(00100011)^{\infty}$  \\
			\cline{1-2}
			(10011000), (00100110), (10001001) &(10011000), (00100110), (10001001) &   & $(10011000)^{\infty}$ \\
			\cline{1-2}
			(01100111), (11011001), (01110110) &(01100111), (11011001), (01110110) &   & $(01100111)^{\infty}$ \\
			\cline{1-2}
			(11011100), (01101110), (11001101) &(11011100), (11011001), (01110110) &   & $(11011100)^{\infty}$ \\
			\cline{1-2}
			(11110000), (00001111)&(11110000), (00001111)&   & $(11110000)^{\infty}$ \\
			\cline{1-2}
			(10110100), (01001011)&(10110100), (01001011)&   & $(10110100)^{\infty}$ \\
			\cline{1-2}
			(00001101)&(00001101)&   & $(00001101)^{\infty}$ \\
			\cline{1-2}
			(00001011)&(00001011)&   & $(00001011)^{\infty}$\\
			\cline{1-2}
			(11110100)&(11110100)&   & $(11110100)^{\infty}$\\
			\cline{1-2}
			(11110100)&(11110100)&   & $(11110100)^{\infty}$\\
			\hline
			(01010000), (00001010)  & (00001010) & \multirow{8}{*}{1} &$(00001010)^{\infty}$\\
			\cline{1-2}
			(10101111), (11110101) &(11110101)&   & $(11110101)^{\infty}$\\
			\cline{1-2}
			(00001110)&(00001110)&   & $(00001110)^{\infty}$\\
			\cline{1-2}
			(11110001)&(11110001)&   & $(11110001)^{\infty}$\\
			\cline{1-2}
			(10101100)&(10101100)&   & $(10101100)^{\infty}$\\
			\cline{1-2}
			(01010011)&(01010011)&   & $(01010011)^{\infty}$\\
			\cline{1-2}
			(11011000)&(11011000)&   & $(11011000)^{\infty}$\\
			\cline{1-2}
			(00100111)&(00100111)&   & $(00100111)^{\infty}$\\
			\hline
			(00001001), (10010000)  & (10010000) & \multirow{6}{*}{2} &$(10010000)^{\infty}$\\
			\cline{1-2}
			(01000101), (01010001) &(01010001)&   & $(01010001)^{\infty}$\\
			\cline{1-2}
			(10111010), (10101110) &(10101110)&   & $(10101110)^{\infty}$\\
			\cline{1-2}
			(11110110), (01101111) &(01101111)&   & $(01101111)^{\infty}$\\
			\cline{1-2}
			(00001100)&(00001100)&   & $(00001100)^{\infty}$\\
			\cline{1-2}
			(11110011)&(11110011)&   & $(11110011)^{\infty}$\\
			\hline
			(00010000), (00001000)  & (00001000) & \multirow{4}{*}{3} &$(00001000)^{\infty}$\\
			\cline{1-2}
			(01010010), (01001010) &(01001010)&   & $(01001010)^{\infty}$\\
			\cline{1-2}
			(10101101), (10110101) &(10110101)&   & $(10110101)^{\infty}$\\
			\cline{1-2}
			(11101111), (11110111) &(11110111)&   & $(11110111)^{\infty}$\\
			\hline
		\end{tabular}
	\end{center}
\end{table}
\end{example}

The previous section discussed the generation of binary de Bruijn sequences, which is a particular case for $\omega \leq \frac{n}{2}$.
Below we discuss another particular case where the prescribed nonlinear complexity $\omega$ achieves the maximum value $n-1$.

\begin{remark}
It was shown in \cite{SZLH} that sequences in $\mathcal{P}(n,n-1)$ can be
generated from a recursive approach by applying the Euclidean algorithm on $n$ and certain positive integers.
Similarly, the set $\mathcal{P}(n, n-2)$ was later completely characterized in \cite{Xiao2018}.
In this paper, Theorem \ref{thm_core} (ii) can produce all sequences in $\mathcal{P}(n, \omega)$ with $\frac{n}{2} \leq \omega \leq n-1$.
Below we discuss the connection between Theorem \ref{thm_core} (ii)  and the work in \cite{SZLH}. Its connection with the work in \cite{Xiao2018} is similar and thus not included here.

For the case $\omega = n-1$ and a sequence $\sn \in \mathcal{R}(n, \nup)$ with $spac(\sn)=d$, namely,
$$
\textbf{s}_{n}
=(\obr{\textbf{s}_d}{q}\, (s_{0}, \dots, s_{r-1}, \overline{s}_{r})\, \s_{[c+d:n]})
=(\textbf{s}_d^{q}\, (s_{0}, \dots, s_{r-1}, \overline{s}_{r})\, \s_{[c+d:n]}),
$$
suppose $add(\sn)=t=\omega-c= \n- 1$, then
$
s_{n-1-i}=s_{(d-1-i)\,\text{mod}\,\,d}\,\, \text{ for } 0\leq i< \n - 1.
$
and $
s_{n-\n}\neq s_{d-\n}.
$
When $\omega+d  \leq n$, i.e., $d=1$, the sequence $\sn = (s_0^c\, \overline{s}_0 \, s_0^t)$. Then
$$
\textbf{a}_n = R^t(\sn)=(s_0^{n-1}\,\overline{s}_0),
$$ which corresponds to the sequence in \cite[Theorem 1 (i)]{SZLH}. When $\omega+d   > n$, suppose $\s_{[i:i+\omega+d]}$ is the subsequence formed by the companion pair
$(\textbf{\underline{s}}_i, \underline{\textbf{s}}_{i+d})$.
Then $\textbf{a}'_n=L^i(\sn)$ has the companion pair
$(\textbf{\underline{s}}_0, \underline{\textbf{s}}_{d})$.			
As discussed in the proof of Theorem \ref{c+t}, we know the sequence $\textbf{a}_n=L^j(\textbf{a}'_n)$ with $j=(\omega+d)-n$
belongs to $\mathcal{B}(n, n-d, d)$ and has $add(\textbf{a}_n)=\omega-(n-d) = d-1$. This implies
that $\textbf{b}_n=R^{d-1}(\textbf{a}_n) $ satisfies $b_i=b_{i+d}$ for $i=0,1, \dots, n-d-1$ and has the form
$$\textbf{b}_n =(\obr{\textbf{b}_d}{q}\,(b_0,\dots, b_{r-1}))=( \textbf{b}_{d}^q \,(b_0,\dots, b_{r-1}))$$
where $r = n\,\text{mod}\,\,d$ and $q = \frac{n-r}{d}$. This corresponds to the sequences characterized in  \cite[Theorem 1 (ii)]{SZLH}. For instance, given a sequence $\sn=(010101001\mathbf{0101010})\in \mathcal{B}(16,8,7)$ with $add(\sn)=7$,
where the added terms are in bold,
we have $\textbf{a}_n =L^{15}(\sn)= (0010101001\mathbf{010101}) \in  \mathcal{B}(16,9,7)$ with $add(\textbf{a}_n)=6$. Then the sequence $\textbf{b}_n=R^{d-1}(\textbf{a}_n)=R^{6}(\textbf{a}_n) =(0101010010101001)$
has the form $(0101010)^2(01) = ((01)^30)^2(01)$, which is consistent with the instance given in \cite[Example 1]{SZLH}.	
\end{remark}

According to Theorem \ref{thm_core} (ii),
each periodic sequence with a prescribed nonlinear complexity $\omega \geq \frac{n}{2} $ can be
derived from its $n$-length subsequence $\sn$ in $\mathcal{R}(n, c)$ with a desired $add(\sn)=\omega-c$ with $c=\lceil \frac{n}{2}\rceil$.
It is thus of significant interest to further explore the structure of sequences in $\mathcal{R}(n, c)$.
The following subsection further discusses the sequence representatives in $\mathcal{R}(n, c)$.
After the discussion, we will propose Algorithm \ref{alg:period-large} to generate all binary sequences in $\mathcal{P}(n,\omega)$ with $\omega \geq \frac{n}{2}$.

\subsubsection{Characterizations on sequence representatives}\label{Sec4.2.1}	
For $\sn\in \mathcal{B}(n, c)$ with %
$c =\nup$
given in (\ref{structure of finite}),
this section further investigates the set $E(\sn)= \{L^{k}(\textbf{s}_{n})\,:\, 0\leq k<n \}\cap \mathcal{B}(n,c)$, which helps us generate sequence representatives in $\mathcal{B}(n,c)$ more efficiently.
In what follows, we shall first consider the necessary conditions such that both $\sn$ and $R^{h}(\sn)$ are contained in $\mathcal{B}(n, c)$.

\begin{lem}\label{chara_h}
For $\textbf{s}_{n}\in \mathcal{B}(n, c, d)$ with $c\geq \frac{n}{2}$, if $\textbf{s}_{n}$ has a shift equivalent sequence $\textbf{v}_n=R^{h}(\textbf{s}_{n})\in \mathcal{B}(n,c, d')$,
then $h$ satisfies at least one of inequalities: $n-(c+d)+d'\leq h <c+d'$ and  $n-(c+d)<h \leq c+d'-d$.
\end{lem}
\begin{proof}
For a sequence $\sn \in \mathcal{B}(n,c,d)$, we shall first prove that if $R^{h}(\sn)$ belongs to $\mathcal{B}(n,c, d')$, then the value $h$  satisfies $n-(c+d)< h < c+d'$.
Suppose that $h \leq n-c-d$. According to (\ref{structure of finite}),
$\sn \in \mathcal{B}(n, c, d)$ has the form
$$
\textbf{s}_{n}
=(\textbf{s}_d^{q}\,(s_{0}, \dots, s_{r-1},  \overline{s}_{r} )\, \s_{[c+d:n]}),
$$
where  $(q-1)d+r+1=c$. It is clear that the $c$-tuples
$(\textbf{s}_d^{q-1}\,(s_{0}, \dots, s_{r-1},  {s}_{r} ))$
and 
$(\textbf{s}_d^{q-1}\,(s_{0}, \dots, s_{r-1},  \overline{s}_{r} ))$
form a companion pair in $\sn$. Due to $h \leq n-(c+d)$, this companion pair is also contained in
the right cyclic shift sequence $	R^{h}(\textbf{s}_{n})$.
Since $nlc(R^{h}(\textbf{s}_{n}))=c\geq \frac{n}{2}$, 
Lemma \ref{lem unique} shows that
this companion pair is the unique one of $R^{h}(\textbf{s}_{n})$.
It contradicts the assumption that $R^{h}(\textbf{s}_{n})$ belongs to $\mathcal{B}(n,c)$. Thus we have $h> n-(c+d)$.
For the sequence $\textbf{v}_n = R^h(\sn)$ and $h'=n-h$,
similarly we have $h'=n-h > n-(c+d')$, implying $h<c+d'$. Thus the statement $n-(c+d)< h < c+d'$ holds.

Furthermore, by setting $b=h-(n-(c+d)))$ and $b'=h'-(n-(c+d'))$, we have $b+b' = (h+h') - (n-(c+d) + n - (c+d')) = (2c-n) + d +d' \geq d + d'$ since $c\geq \frac{n}{2}$.
Hence at least one of $b \geq d'$ and $b' \geq d$ holds.
If $b \geq d'$, then $n-(c+d)+d'\leq h < c+d'$;
if $b' \geq d$, then $ n-(c+d')+d \leq h'=n-h< c+d$ and thus $ n-(c+d)<h\leq (c+d') -d$.	
\end{proof}

In the proof of Lemma \ref{chara_h}, we see the symmetric relation between $\sn$ and $\textbf{v}_n=R^{h}(\textbf{s}_{n})$ in $\mathcal{B}(n, c)$.
If a pair of shift equivalent sequences $(\textbf{s}_{n}, R^{h}(\textbf{s}_{n}))$ in $\mathcal{B}(n,c)$ satisfies  $h\leq (c+d')-d$, then
$(\textbf{v}_n=R^{h}(\textbf{s}_{n}),R^{n-h}(\textbf{v}_n) )$ satisfies $h'=n-h\geq (n-(c+d'))+ d$ 
where $spac(\textbf{s}_n)=d$ and $spac(\textbf{v}_n)=d'$.
Therefore,
it suffices to consider the shift equivalent sequence $R^h(\sn) $ of each $\textbf{s}_{n}$ in $\mathcal{B}(n,c)$ with $ (n-(c+d))+ d' \leq h\leq c+d'-1$.

In the following,
for the pair of shift equivalent sequences $(\textbf{s}_{n},\textbf{v}_{n})$ in $\mathcal{B}(n, c)$,
Proposition \ref{arbitrary seq}
characterizes certain structure of the sequence $\sn$ in (i) and determines $\textbf{v}_{n}$ and $add(\textbf{v}_{n})$ in (ii),
which will help us find shift equivalent sequences and then determine the sequence representatives
by deleting sequences that are not sequence representatives in Algorithm \ref{alg:period-large}.

\begin{prop}\label{arbitrary seq}	
Suppose a sequence $\textbf{s}_{n}\in \mathcal{B}(n, c, d)$ with $c=\nup$  has a shift equivalent sequence $\textbf{v}_n=R^{h}(\textbf{s}_{n})\in \mathcal{B}(n,c, d')$, where $(n-(c+d))+ d' \leq h\leq c+d'-1$.
Then

\noindent{\rm (i)}
the subsequence $\textbf{s}_{[g:n]}$ of $\textbf{s}_{n}$ satisfies
\begin{equation}\label{s-n-sat}
	s_i = s_{i+d'} \ \text{ for } \ g\leq i < n-d',
\end{equation}
where $g=(c+d)-d'$ in the case of $d<n-c$, and $g=2d-d'-1$ in the case of $d=n-c$;				

\noindent{\rm (ii)} 	
let $r_1\leq g$ and $r_2\geq n-d'$.
Suppose
$s_i=s_{i+d'}
$ holds for any integer $r_1\leq i < r_2$ and $s_i\neq s_{i+d'}$ for $i=r_1-1, r_2$, where the subscripts are taken modulo $n$. Take $c'=r_2-r_1+1$.
If $c'\geq c$,  then $h= c + (n-r_2-1) $ and $add(\textbf{v}_n)=c'-c$.	
\end{prop}
\begin{proof}	
{\rm (i)}
Since $n-(c+d)+d' \leq h < c+d'$,
$\s_{[(c+d)-d':n]} = (s_{(c+d)-d'},\dots,s_{n-1})$ is a subsequence of $\s_{[n-h:n]}$, and $\s_{[n-h:n]}$ is contained in the subsequence $\textbf{v}_{c+d'-1}$ of
$\textbf{v}_n = R^h(\sn)=\s_{[n-h:n]}\s_{n-h} $ in $\mathcal{B}(n, c, d')$. According to the definition of $\mathcal{B}(n, c, d')$ in (\ref{structure of finite}), the subsequence $\textbf{v}_{c+d'-1}$ satisfies
$v_i = v_{i+d'}$ for  $i=0, 1, \dots, c-2$. Thus it follows from $\s_{[(c+d)-d':n]}
\subseteq \s_{[n-h:n]}\subseteq \textbf{v}_{c+d'-1}   $
that
$
s_{i} = s_{i+d'}
$ for $i=(c+d)-d', \dots, n-d'-1$. Here it requires $d<n-c$ from $(c+d)-d'\leq n-d'-1$. That is, 
in the case of  $d<n-c$, the subsequence $\textbf{s}_{[(c+d)-d':n]}$ of $\textbf{s}_{n}$ satisfies
\begin{equation}\label{eq_d_small}
	s_i = s_{i+d'} \ \text{ for } \ (c+d)-d'\leq i < n-d'.
\end{equation}

In the case of $d= n-c$:
consider $\textbf{v}_n\in \mathcal{B}(n,c, d')$ and $R^{h'}(\textbf{v}_{n})=\textbf{s}_n$ with $h'=n-h$.
We divide the discussion into two subcases according to the value of $h'$.
For the subcase of $h' \geq [n-(c+d')]+d$, then $h\leq (c+d')-d$.
When $d'<n-c$,  it follows from \eqref{eq_d_small} that
$\textbf{v}_n \in \mathcal{B}(n,c, d')$ satisfies
$	v_i = v_{i+d}$ with $(c+d')-d\leq i < n-d$.
Thus the pair of shift equivalent sequences $\textbf{v}_n$ and $R^{h'}(\textbf{v}_{n})=\sn$ has been considered in \eqref{eq_d_small}.
That is, when $d'<n-c$ there is no need to consider $\sn$ and $R^{h}(\sn)=\textbf{v}_n$.
When $d'=n-c$, $(n-(c+d))+ d' \leq  h\leq (c+d')-d$ implies that $n-c \leq h \leq c$.	
If $c=\frac{n}{2}$ with even $n$, we have
$\sn=(\s_{d}\,(s_0,\dots,s_{d-2},\overline{s}_{d-1}))$  
and $h=\frac{n}{2}$.
If $(s_0,\dots,s_{d-2},\overline{s}_{d-1})$ is aperiodic, then $\sn$ and $R^{\frac{n}{2}}(\sn)$ are shift equivalent in $\mathcal{B}(n,\frac{n}{2})$ with the same number of added terms. In this time, $\sn$ can be kept as a candidate for the sequence representatives.
If $c=k+1$ with odd $n=2k+1$, we have 
$\sn=(\s_{k}^2 \,\overline{s}_{0})$
and $k \leq h\leq k+1$,
while $R^{k}(\sn)$ and $R^{k+1}(\sn)$ do not belong to $\mathcal{B}(2k+1,k)$.
That is to say, when $d= n-c$ the subcase of $h' \geq [n-(c+d')]+d$ does not need to be considered.

For the subcase of $h'< [n-(c+d')]+d$, then $h> (c+d')-d$, implying that $ (c+d')+1-d \leq h \leq c+d'-1$. 
Note that $n-((c+d')+1-d)=2d-d'-1$.
Thus $\s_{[2d-d'-1:n]} = (s_{2d-d'-1},\dots,s_{n-1})$ is a subsequence of $\s_{[n-h:n]}$, and $\s_{[n-h:n]}$ is contained in the subsequence $\textbf{v}_{c+d'-1}$ of
$\textbf{v}_n = R^h(\sn) $ in $\mathcal{B}(n, c, d')$.
According to the definition of $\mathcal{B}(n, c, d')$ in (\ref{structure of finite}), we have the subsequence $\s_{[2d-d'-1:n]}$ of $\textbf{s}_{n}$ satisfies
\begin{equation*}
	s_i = s_{i+d'} \ \text{ for } \ 2d-d'-1\leq i < n-d'.
\end{equation*}
The desired conclusion in \eqref{s-n-sat} thus follows.	

{\rm (ii)} 	
With the assumption that
$
s_i = s_{i+d'}$ holds for $r_1 \leq i <  r_2$, 
the sequence
$\textbf{u}_n^{\infty} =(u_0, u_1, \dots, u_{n-1})^{\infty}= L^{r_1}(\sni)$  satisfies
\begin{equation}\label{relast}
	u_i = u_{i+d'}\, \text{ for }\, 0\leq i < r_2-r_1 \, \text{ and } \, u_i \neq u_{i+d'} \, \text{ for } \, i= -1,  r_2-r_1.
\end{equation}
This implies that $\textbf{u}_{c'+d'}$ belongs to $\mathcal{B}(c'+d', c', d')$ by $c' = r_2-r_1+1$.
When $c'+d' \leq n$, from \eqref{relast} we have $\textbf{u}_{n} \in \mathcal{B}(n, c', d')$ with $add(\textbf{u}_n) = 0$. Thus, if $c' \geq c$, then by Definition \ref{def-t} and Lemma \ref{c-c+1}
we have the sequence $\textbf{v}_{n} = L^{c'-c}(\textbf{u}_{n})\in\mathcal{B}(n, c, d')$ with $add(\textbf{v}_{n}) = c'-c$.
When $c'+d' > n$, let $\textbf{u}'_n=\textbf{u}_{[c'+d'-n,c'+d']}=L^{c'+d'-n}(\textbf{u}_n)$. It follows from Lemma \ref{lem_s_i^i+k+w} (iii) and \eqref{relast} that $\textbf{u}'_n \in \mathcal{B}(n, n-d', d')$ with $add(\textbf{u}'_n) = c'+d'-n$.
Due to $c+d'\leq n$,
again by Definition \ref{def-t} and Lemma \ref{c-c+1} it implies that $L^{n-d'-c}(\textbf{u}'_{n})=L^{c'-c}(\textbf{u}_{n})=\textbf{v}_{n}$ belongs to $\mathcal{B}(n, c, d')$ with $add(\textbf{v}_{n})=(c'+d'-n)+(n-d'-c)=c'-c$.
Together the two cases,
we both have
$\textbf{v}_n = L^{c'-c}(\textbf{u}_n)  = L^{r_1+(c'-c)}(\textbf{s}_n)= R^{h}(\sn)$ with $h = n-(r_1+c'-c)= c + (n-r_2-1) $ and $add(\textbf{v}_{n})=c'-c$.		
\end{proof}

In order to generate the sequence representatives in the set $\mathcal{R}(n,\nup)$,
from Proposition \ref{arbitrary seq} (i), we only need to investigate sequences $\textbf{s}_{n}$ in $\mathcal{B}(n,\nup)$ satisfying (\ref{s-n-sat}) for any
$1\leq d'\leq \n$,
since only these sequences may have shift equivalent sequences in $\mathcal{B}(n,\nup)$.
Thus we can generate the set $\mathcal{R}(n,c)$ with $c=\nup$ as follows.
For each $1\leq d\leq n-c$ and each $\textbf{s}_{c+d}$,
consider $\sn$ satisfying \eqref{s-n-sat} for every aperiodic subsequence $(s_{c+d-d'},\dots,s_{c+d-1})$ with $1\leq d'\leq n-c$.
By Proposition \ref{arbitrary seq} (ii), suppose that $\sn$ satisfies
$c' \geq c$, then
we obtain directly $\textbf{v}_n=R^h(\sn) \in \mathcal{B}(n, c)$ and $add(\textbf{v}_n)= c'-c $. 
When $\sn$ and $R^{h}(\textbf{s}_{n})=\textbf{v}_n$ are considered as a pair of shift equivalent sequences in $E(\sn)$,
the one with less added terms will be deleted.
As we let all subsequences $\textbf{s}_{c+d}$ go through the above process, every pair of shift equivalent sequences
in $E(\sn)$ are considered. Thus all sequences $\textbf{a}_n$ with maximal $add(\textbf{a}_n)$ are kept.
This allows us to obtain sequence representatives more efficiently, which in turn constitute the set $\mathcal{R}(n,c)$. 
Below we give an example to illustrate the above analysis to delete sequences which are not representatives by Proposition \ref{arbitrary seq}.

\begin{algorithm}
	\footnotesize
	\caption{Generation of all periodic binary sequences in $ \mathcal{P}(n,\omega), \omega\geq  \frac{n}{2}$ }\label{alg:period-large}
	\begin{algorithmic}[1]
		\State\textbf{Main Algorithm}
		\State    INPUT: A positive integer $n$ 
		\State    OUTPUT: The set $\mathcal{P}(n,\omega)$ with $ \nup \leq \omega\leq n-1$.
		\State $\mathcal{R}(n, \nup)  \gets \mathtt{genR}(n, \nup)$
		\State  $\mathcal{P}(n, \omega)=\left\{(L^k(\textbf{s}_n))^{\infty} \,: \sn \in \mathcal{R}(n,\nup),\, add(\sn)=\omega-\nup, 0\leq k<n \right\}$
		\Function {$\mathtt{genR}$}{$n,\nup$}         \hfill// Generate $\mathcal{R}(n, \nup)$ 
		\State {$\mathcal{R}(n,c)\gets \emptyset$, $U\gets \emptyset$ and $c\gets \nup$}
		\For {$d=1$ to $n-c$}
		\While {$(v_{0},v_1,\dots,v_{d-1})\in\mathbb{Z}_2^{d}$ is aperiodic}
		\State {$s_{i}\gets v_{i\, \text{mod}\,\,d}, \  0\leq i\leq c+d-2$, $s_{c+d-1}\gets v_{(c+d-1)\,\text{mod}\,\,d\,}\oplus1$ and $V\gets \emptyset$ }
		\If {$d<n-c$}
		\For {$d'=1$ to $n-c$}
		\If {$(s_{c+d-d'},\dots,s_{c+d-1})$ is aperiodic}
		\State {$s_{i}\gets s_{i-d'}, \  c+d\leq i\leq n-1$ and $r_1\gets c+d-d'$}
		\If {$\sn= ({s}_{0},s_{1},\dots,s_{n-1}) \notin U$}
		\While{$s_{r_1-1}=s_{r_1+d'-1}$}
		\State {$r_1=r_1-1$}
		\EndWhile
		\State  Similarly obtain $r_2$ such that $s_i=s_{i+d'}$ holds for $ n-d'-1\leq i < r_2$
		\If{$\Delta=(r_2-r_1+1)-c\geq 0$}  \hfill//Exclude sequences from $\mathcal{R}(n,c)$
		\State{If $add(\sn)< \Delta $, then add $(s_{c+d},\dots,s_{n-1})$ to $V$}
		\State{Otherwise set $h=c+(n-r_2-1)$ and add $R^{h}(\textbf{s}_{n})$ to $U$}          
		\EndIf
		\EndIf
		\EndIf
		\EndFor
		\State Add $\textbf{u}_{n}\gets (\textbf{s}_{c+d},u_{c+d}, \dots, u_{n-1})$ with $(u_{c+d},\dots, u_{n-1})\in\mathbb{Z}_2^{n-c-d} \backslash V$ to $\mathcal{R}(n,c)$
		\ElsIf {$\sn=\textbf{s}_{c+d}= ({s}_{0},s_{1},\dots,s_{c+d-1}) \notin U$}
		\State {Add $\sn$ to $\mathcal{R}(n,c)$.}
		\For {$d'=1$ to $n-c$}
		\If{ $(s_{n-d'},\dots,s_{n-1})$ is aperiodic, $s_i=s_{i+d'}$ with $i \in [2d-1-d' , n-d')$}
		\State{Repeat Steps 16$-$23 while revising Step 21 as}
		\State{If $add(\sn)< \Delta$ then add $\sn$ to $U$.}
		\EndIf
		\EndFor
		\EndIf
		\EndWhile
		\EndFor
		\State \Return $\mathcal{R}(n,c)\gets \mathcal{R}(n,c) \backslash U$
		\EndFunction
	\end{algorithmic}
\end{algorithm}

\begin{table}[!ht]
\renewcommand{\arraystretch}{1.5}
\scriptsize
\caption{The shift equivalent sequences for $n=12$, $c=6$ and $d=2$ with $\textbf{s}_{d}=(01)$.}\label{tab_3}
\begin{center}
	\begin{tabular}{|c|c|c|c|c|c|c| p{1cm}|}
		\hline
		\multicolumn{8}{|c|}{$\textbf{s}_{c+d}=(01010100)$  \quad and \quad $\textbf{v}_{n}=R^{h}(\textbf{s}_{n}) \in \mathcal{B}(n,c,d')$ } \\
		\cline{1-8}
		$d'$&   $\textbf{s}_{n}$ & $r_1$ & $r_2$ & $add(\textbf{v}_{n})=c'-c$  & $h$ & $\textbf{v}_{n}$ & Delete  \\
		\hline
		\hline
		1 & $(\underline{0}10101\underline{\mathbf{0}}00000)$  & 6&	12 & 	1 & 5 &    $(000000101010)$&   $\textbf{s}_{n}$  \\
		\hline
		\multicolumn{1}{|c}{2}  & \multicolumn{7}{|c|}{--}  \\
		\hline
		3 & $(0101\underline{\mathbf{0}}10010\underline{0}1)$  & 4&	10 & 	1 & 7 &    $(100100101010 )$  &   $\textbf{v}_{n}$\\
		\hline
		4 & $(0101\underline{\mathbf{0}}100010\underline{0})$ & 4&	11 & 	2 & 6 &    $(000100010101 )$  &   $\textbf{s}_{n}$\\
		\hline
		\multirow{1}{*}{5} & \multirow{1}{*}{$(\underline{0}1\underline{\mathbf{0}}101001010)$} &	\multirow{1}{*}{2} & 	\multirow{1}{*}{12}& \multirow{1}{*}{5} & \multirow{1}{*}{5} &   \multirow{1}{*}{$(010100101010)$}  &   \multirow{1}{*}{$\textbf{s}_{n}$} \\
		\hline
		6 & $(01\underline{\mathbf{0}}1010\underline{0}0101)$ &2&	7 & 	0 & 10 &    $(010100010101 )$&   $\textbf{v}_{n}$  \\
		\hline
	\end{tabular}
\end{center}
\end{table}
\begin{example}
Take an example for $n=12$, $c=6$ and $d=2$.
Consider $\textbf{s}_{d}=\textbf{s}_{2}=(01)$
and $\textbf{s}_{c+d}=(\mathbf{01}010100)$. 
Run through $d'$ with $1\leq d'\leq 6$.
For $d'=1$, by \eqref{s-n-sat} only the subsequence $\textbf{s}_{[8:12]}=(0000)$ needs to be considered, thus $\textbf{s}_{12}=(\underline{0}10101\underline{\mathbf{0}}00000)$.
By Proposition \ref{arbitrary seq} (ii), we can see that $r_1=6$ (the corresponding terms given in bold and underlined), and
$r_2=12$ (the corresponding terms are underlined),
implying $c'=r_2-r_1+1=7$.
Since $c'>c$, we get $h=n-(r_2+1-c)=5$. So $\textbf{v}_{12}=R^{h}(\textbf{s}_{12})=R^{5}(\textbf{s}_{12})\in \mathcal{B}(12, 6,1)$ with $add(\textbf{v}_{12})=c'-c=1$.
The pair of shift equivalent sequences $\textbf{s}_{12}$ and $\textbf{v}_{12}$ is found, and $\textbf{s}_{12}$ should be deleted since $add(\textbf{s}_{12})=0<add(\textbf{v}_{12})=1$.
Since the last two terms of $\textbf{s}_{c+d}=(01010100)$ is $(00)$, it is impossible for $d'=2$ by the definition of aperiodic sequences.
The sequences which are not representatives can be deleted similarly for $d'=3,4,5,6$, which are given in Table \ref{tab_3}.
\end{example}

\begin{remark}
In Algorithm \ref{alg:period-large}, from Proposition \ref{arbitrary seq} we give the detailed steps to generate the set $\mathcal{R}(n, c)$ for any $c= \nup$.
For $\nup \leq \omega \leq n-1$, Algorithm \ref{alg:period-large} generates all periodic sequences in $\mathcal{P}(n,\omega)$ based on Theorem \ref{thm_core} (ii).
From the steps in generating $\mathcal{R}(n, c)$, we see that the loops on $(v_0, \dots, v_{d-1})$ and $(u_{c+d}, \dots, u_{n-1})$ contribute  the dominating time and memory complexity, roughly, $O(2^{n/2})$.
This indicates that  Algorithm \ref{alg:period-large} has an advantage of factor $2^{n/2} n \log_2(n)$ compared to the exhaustive search for sequences $\sni$ with nonlinear complexity $\omega$.

\end{remark}

\section{Conclusion}\label{Sec6}
Our contributions in this paper are twofold:
the first contribution is to investigate the varying behavior of the nonlinear complexity of finite-length sequences under circular shift operators, and
the second contribution is the establishment of a one-to-one correspondence between the set of periodic sequences with certain nonlinear complexities and
the set of certain finite-length sequences with a particular structure.
As an application of the correspondence, we present two efficient algorithms to generate all periodic sequences with
any prescribed nonlinear complexity.

\section*{Acknowledgments}
The authors would like to thank the editor and the
anonymous referees for their detailed reading and valuable
comments that greatly improved the presentation and quality
of the article.

\end{document}